\documentclass{amsart}

\usepackage[lite]{amsrefs}
\usepackage{amssymb}
\usepackage[all,cmtip]{xy}
\usepackage{amsthm}
\usepackage{amsmath}
\usepackage{amsfonts}
\usepackage[hidelinks]{hyperref}
\usepackage{enumerate}
\usepackage{physics}
\usepackage[titletoc]{appendix}
\usepackage{fancyhdr}
\usepackage{xcolor}
\usepackage{comment}
\usepackage{tikz-cd}
\usepackage{tikzpagenodes}
\usepackage{caption}
\usepackage{float}
\usetikzlibrary[patterns]
\usepackage[margin = 1in]{geometry}
\usepackage{setspace}

\spacing{1.2}

\allowdisplaybreaks

\numberwithin{equation}{section}
\newtheorem{definition}{Definition}[section]
\newtheorem{example}[definition]{Example}
\newtheorem{theorem}[definition]{Theorem}

\newtheorem{corollary}[definition]{Corollary}
\newtheorem{proposition}[definition]{Proposition}
\newtheorem{remark}[definition]{Remark}

\newtheorem{question}[definition]{Question}

\begin{document}
\title[Supersymmetric extension of universal enveloping vertex algebras]{Supersymmetric extension of universal enveloping \\ vertex algebras}

\author[U.R.Suh]{Uhi Rinn Suh}
\address{Department of Mathematical Sciences and Research institute of Mathematics, Seoul National University, Gwanak-ro 1, Gwanak-gu, Seoul 08826, Korea}
\email{uhrisu1@snu.ac.kr}

\author[S.Yoon]{Sangwon Yoon}
\address{Department of Mathematical Sciences, Seoul National University, Gwanak-ro 1, Gwanak-gu, Seoul 08826, Korea}
\email{ysw317@snu.ac.kr}

\thanks{This work was supported by NRF Grant, \#2022R1C1C1008698 and Creative-Pioneering Researchers Program through Seoul National University. Section 4 of this paper is based on the master thesis of S. Yoon.}

\begin{abstract}
In this paper, we study the construction of the supersymmetric extensions of vertex algebras. In particular, for $N = n \in \mathbb{Z}_{+}$, we show the universal enveloping $N = n$ supersymmetric (SUSY) vertex algebra of an $N = n$ SUSY Lie conformal algebra can be extended to an $N = n' > n$ SUSY vertex algebra.
\end{abstract}

\maketitle

\vskip 6mm

\section{Introduction}\label{section:Intro}
\setcounter{equation}{0}

Vertex algebras were introduced by R. Borcherds in the 1980's (see \cite{Bor86}) to describe the chiral part of two-dimensional conformal field theory (see \cites{BPZ84}). A vertex algebra consists of a vector superspace $V$, $\text{End}(V)$-valued quantum fields $a(z)$ for $a \in V$ and a translation operator $\partial$ with certain axioms. One of the axiom called locality tells that, for a given pair of fields $(a,b)$ in $V$, there is a positive integer $t \in \mathbb{N}$ such that $(z-w)^{t}[a(z), b(w)]=0$. Considering the formal delta distribution $\delta(z,w):= \sum_{j \in \mathbb{Z}} z^j w^{-j-1}$, the elements $a_{(m)}b \in V$ for $m = 0,1,\cdots, t$ are determined by the supercommutator formula: 
\begin{equation} \label{eq:decomposition theorem}
 [a(z), b(w)] =\sum_{m = 0}^{t} \frac{1}{m!}a_{(m)}b(w) \ \partial_w^m\delta(z,w).
\end{equation}
The existence of such formula \eqref{eq:decomposition theorem} is guaranteed by so-called the decomposition theorem. Additionally, \eqref{eq:decomposition theorem} and the state-field correspondence yield another well-known definition of a vertex algebra using the $\lambda$-bracket formalism (see Section \ref{section2}). Speaking briefly, a vertex algebra $V$ is a differential quasi-commutative quasi-associative supereralgebra endowed with a $\lambda$-bracket which is compatible with the differential algebra structure. Here, the derivation on $V$ for the differential algebra structure arises from the translation operator $\partial.$

To describe two-dimensional conformal field theory with supersymmetries in mathematical language, it is good to use the notion of supersymmetric vertex algebra, which is a generalization of the superconformal vertex algebras (see Section 5.9 of \cite{Kac98}). A supersymmetric vertex algebra consists of a vector superspace $V$ and superfields $a(z, \Theta):=a(z,\theta_1,\theta_2, \cdots, \theta_n)$ for $a\in V$ and odd variables $\theta_i$. Roughly speaking, the superfield
\begin{equation*} \label{eq:superfield}
a(z, {\Theta})=\sum_{I} \pm \ \theta^I D^I a(z),
\end{equation*}
where the summation is over the ordered subset $I=\{i_1,i_2, \cdots, i_s\}\in \{1,2,\cdots, n\}$ satisfying $i_1<i_2<\cdots<i_s$, $D^I a(z):=D^{i_1}D^{i_2}\cdots, D^{i_s}a(z)$ is a field of parity $p(a)+\bar{s}$
and $\theta_I= \theta_{i_1}\theta_{i_2}\cdots \theta_{i_s}$. Here $D^i$'s can be understood as odd endomorphisms of fields such that $D^iD^j+D^jD^i= \delta_{i,j} 2 \partial$ for $i,j\in I$
and, indeed, the odd endomorphisms determine the structures of $V$ coming from the supersymmetry. In particular, in $N=1$ (resp. $N=2$) case, a superfield $a(z, \Theta)$ consists of a couple (resp. quadruple) of fields $(a(z), Da(z))$ (resp. $(a(z), D^1a(z), D^2 a(z), D^1 D^2(z))$). Analogous to the $\lambda$-bracket formalism in vertex algebras theory, an $N=n$ supersymmetric vertex algebra can be introduced via so-called $\Lambda=(\lambda, \chi^1, \chi^2, \cdots, \chi^n)$-bracket formalism (see Section \ref{section3}).

As one can expect, a SUSY vertex algebra is a vertex algebra by itself but the converse cannot be true. However, it is well-known that if a vertex algebra has a superconformal vector then the vector induces a supersymmetry in the vertex algebra (see, for example \cites{HK07, Kac98}). Here a superconformal vector $G$ in a vertex algebra $V$ is a primary of conformal weight ${\frac{3}{2}}$ element satisfying
\begin{equation*}
[G \ _\lambda \ G]= 2{L} + {\frac{c}{3}}{\lambda}^{2}
\end{equation*}
for a  conformal vector $L$ of $V$. Indeed, in this case, we can define an odd operator $D$ by letting $D(a)$ to be the coefficient of $\lambda^0$ in $[G \ _\lambda \ a]$ and the vertex algebra $V$ together with $D$ becomes an $N=1$ SUSY vertex algebra. We remark that constructions of $N = 1$ and $N = 2$ superconformal structures in vertex algebras have been studied in various contexts (see \cites{HZ11, JF21}).

In Section \ref{section:sconf str of VA}, we investigate superconformal vectors of a vertex algebra which are main sources of the $N = 1$ SUSY vertex algebra structures. Sometimes, there are more than one superconformal vectors which yield the same SUSY structure but different conformal structures. In this case, we call one of the superconformal vectors a shift of the other one (see Proposition \ref{proposition:shifted sconf}).
The shift of conformal vectors has a crucial application in the context of $W$–algebras (see \cites{Ara15, KRW03}). More precisely, a conformal vector of the $W$–algebra is induced from a conformal vector of the corresponding BRST complex, which is the tensor product of the affine vertex algebra and the free fermion vertex algebra. Hence the sum of conformal vectors of these two algebras becomes a conformal vector of the BRST complex. Each of the affine vertex algebra and the free fermion vertex algebra has an infinite number of conformal vectors, and any two of these vectors are shifts of each other. As a consequence, the BRST complex has infinitely many conformal vectors which are shifts of each other, and exactly one of the them yields a desired conformal structure of the $W$–algebra. In this paper, we present analogous results on the $N = 1$ SUSY BRST complexes (see Theorem \ref{theorem:sconf of tensor of aff and fermion}).

If a vertex algebra does not have a superconformal vector, we cannot guarantee the existence of supersymmetries.
For example, the Virasoro vertex algebra and affine vertex algebras do not have neither superconformal vectors nor known supersymmetries. Nevertheless, we can impose supersymmtries on these vertex algebras by embedding them into bigger supersymmetric vertex algebras.
For the Virasoro vertex algebra, consider the $N=1$ SUSY vertex algebra generated by the odd element $G$ and endowed with the $\Lambda$-bracket:
\begin{equation*} \label{eq:intro-superconformal}
[G \ _\Lambda \ G]= (2\partial+ 3\lambda+\chi D) G + \frac{c}{3}\lambda^2 \chi.
\end{equation*}
This SUSY vertex algebra is called the $N = 1$ super-Virasoro vertex algebra and the vertex subalgebra generated by the even element $D(G)$ is the Virasoro vertex algebra. Hence the  $N = 1$ super-Virasoro vertex algebra can be considered as a supersymmetric extension of the Virasoro vertex algebra.

Similarly, the affine vertex algebra $V^k(\mathfrak{g})$ of level $k$ associated with a Lie superalgebra $\mathfrak{g}$ can be embedded into an $N=1$ SUSY vertex algebra $V^k_{N=1}(\mathfrak{g})$ called the $N = 1$ SUSY affine vertex algebra of level $k$. As a SUSY vertex algebra, ${V}^{k}_{N=1}(\mathfrak{g})$ is generated by the parity reversed vector superspace $\bar{\mathfrak{g}}$ of $\mathfrak{g}$ and the $\Lambda$-bracket is given by:
\begin{equation*}
[\bar{a} \ _\Lambda \ \bar{b}] = (-1)^{p(a)} \big( \overline{[a,b]}+ k\chi(a|b) \big)
\end{equation*}
for $a,b\in \mathfrak{g}$. Now, the vertex subalgebra generated by ${D}{\bar{\mathfrak{g}}}$ is isomorphic to the affine vertex algebra. In addition, in various articles \cites{HK08, MRS21}, supersymmetric extensions of well-known vertex algebras such as lattice vertex algebras and $W$–algebras have been discussed.

In Section \ref{section5}, we investigate where supersymmetries of vertex algebras come from. In short, if a vertex algebra $V$ has an odd derivation $D$ which induces the sesquilinearity of the $\Lambda$-bracket and satisfies $[D,D]= 2\partial$ then 
the operator $D$ gives a supersymmtery in $V$. In terms of $\lambda$-bracket, the sesquilinearity of $\Lambda$-bracket can be rewritten as (SEF1) and (SEF2) in Section \ref{section5} and they are necessary and sufficient conditions for a supersymmetry of a vertex algebra (see Proposition \ref{proposition:N=1 extension strongly generated}). Using these observation, we show the main theorem in Section \ref{section5}:

\begin{theorem}[Theorem \ref{theorem:N=1 extension}] \label{thm:intro-main theorem}
Let $R = \mathbb{C}[\partial] \otimes \mathcal{U}$ for a vector superspace $\mathcal{U}$ be a Lie conformal algebra and $V(R)$ be the universal enveloping vertex algebra of the Lie conformal algebra $R$. The differential algebra $V(\bar{R} \oplus R)$, where $\bar{R}$ is the parity reversed $\mathbb{C}[\partial]$-module of $R$, endowed with the $\Lambda$-bracket 
\begin{equation} \label{eq:main theorem, bracket}
[\bar{a} \ {}_\Lambda \ \bar{b}] := (-1)^{p(a)}\overline{[a \ {}_\lambda \  b]}
\end{equation}
and the odd derivation $D$ such that 
\[ D(\bar{a})=a \quad \text{ and } \quad D(a) = {\partial}{\bar{a}} \]
for $a,b\in \mathcal{U}$ is an $N=1$ SUSY vertex algebra having $V(R)$ as a vertex subalgebra. {\textup{(}}See the picture below.{\textup{)}}
\end{theorem}

\begin{figure}[H]
\centering
\begin{tikzpicture}[scale = 0.9, style={font=\small}, circl/.style={circle, draw = black, fill = black, inner sep = 2pt}, smallcircl/.style={circle, draw = black, fill = black, inner sep = 1pt}, ghostcircl/.style={circle, draw = black, inner sep = 2pt}]\label{fig:N=1 SUSY in LCA}
\node[draw=none, label = above : ${\bar{R}}$] at (1.5, 6) {};
\node[draw=none, label = above : ${R}$] at (6.5, 6) {};
\node[draw=none, label = above : ${V(\bar{R} \oplus R)}$] at (4, 6) {};
\draw[rounded corners=10, thick] (-1, -1) rectangle (9, 7);
\draw[dashed, rounded corners=10, thick] (0, 0) rectangle (3, 6);
\draw[rounded corners=10, thick] (5,0) rectangle (8, 6);
    \node[circl, label = above : {${u}_{i_1} \quad$}] at (6, 5.1) (S1) {};
    \node[circl, label = left : {${\bar{u}}_{j_1}$}] at (2.4, 5.6) (P1) {};
    \node[circl, label = below : {$\quad [{u}_{i_1} \ _{\lambda} \ {\bar{u}}_{j_1}]$}] at (0.7, 4) (T1) {};
\path[shorten > = 7pt, ->] (S1) edge[out = 160, in = 10] (T1);
\path[shorten > = 7pt, -] (P1) edge[out = -30, in = 10] (T1);
    \node[ghostcircl, label = above : {${\bar{u}}_{i_1} \quad$}] at (1, 5.1) (GS1) {};
    \node[ghostcircl, label = left : {${u}_{j_1}$}] at (7.4, 5.6) (GP1) {};
    \node[ghostcircl, label = below : {$\quad [{u}_{i_1} \ _{\lambda} \ {u}_{j_1}]$}] at (5.7, 4) (GT1) {};
\path[dashed, shorten > = 7pt, ->] (S1) edge[out = -20, in = 10, looseness = 1.8] (GT1);
\path[dashed, shorten > = 7pt, -] (GP1) edge[out = -30, in = 10] (GT1);
    \node[draw=none] at (3.2, 3) (A1) {};
    \node[draw=none] at (4.8, 3) (A2) {};
\draw [->] (A1) -- (A2) node[midway, above] {$D$};
    \node[smallcircl, label = below : {${\bar{u}}_{i_2}$}] at (0.5, 2) (S2) {};
    \node[smallcircl, label = below : {${\bar{u}}_{j_2}$}] at (2, 1.3) (P2) {};
\path[dotted, -] (S2) edge[bend left] (P2);
    \node[smallcircl, label = below : {${u}_{i_3} \ $}] at (5.6, 1.2) (S3) {};
    \node[smallcircl, label = below : {${u}_{j_3} \ $}] at (6, 0.6) (P3) {};
    \node[smallcircl, label = above : {$[{u}_{i_3} \ _{\lambda} \ {u}_{j_3}]$}] at (7.1, 2.3) (T3) {};
\path[shorten > = 7pt, ->] (S3) edge[out = -10, in = -100] (T3);
\path[shorten > = 7pt, -] (P3) edge[out = 15, in = -100] (T3);
\end{tikzpicture}
\caption*{Construction of $N = 1$ supersymmetry in Theorem \ref{theorem:N=1 extension}}
\end{figure}
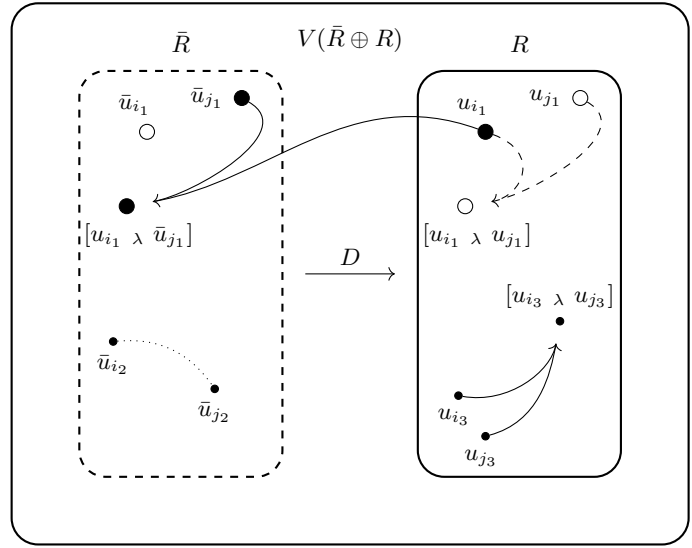

The conclusion of Theorem \ref{thm:intro-main theorem} is, for the universal enveloping vertex algebra of a Lie conformal algebra, one can always find a supersymmetric extension of the vertex algebra. Note that the $\Lambda$-bracket \eqref{eq:main theorem, bracket} is inspired from the bracket of $V^0_{N=1}(\mathfrak{g})$, the SUSY extension of $V^0(\mathfrak{g})$.
However, it is not the only supersymmetric extension of such vertex algebras. For example, the $bc$–$\beta\gamma$ system is a well-known supersymmetric extension of the $\beta\gamma$ system but the SUSY vertex algebra obtained by Theorem \ref{thm:intro-main theorem} is not the $bc$–$\beta\gamma$ system (see Example \ref{ex: betagamma N=1 ext-2}). For the Virasoro vertex algebra, the supsersymmetric extension via Theorem \ref{thm:intro-main theorem} is not the $ N =1$ super-Virasoro vertex algebra (see Example \ref{ex:N=1 ext of vir}).

Moreover, by using the SUSY extension formulas (SEF1) and (SEF2) in Section \ref{section5}, it is possible to impose desired constraints and thereby construct SUSY extensions of a given vertex algebra.
As an example (see Remark \ref{remark:superconformal current}), we can reproduce the operator product expansion of the superconformal current algebra, which was defined in \cite{KT85}. In other words, the $N = 1$ SUSY affine vertex algebra $V^{0}_{N=1}(\mathfrak{g})$ of level $0$ can be embedded into an $N = 1$ superconformal vertex algebra ${\textup{SVir}}^{c} \otimes {V}^{0}_{N=1}(\mathfrak{g})$, where $\textup{SVir}^{c}$ is the $N = 1$ super-Virasoro vertex algebra of central charge $c \in \mathbb{C}$.
Note that $V^k_{N=1}(\mathfrak{g})$ for a basic Lie superalgebra $\mathfrak{g}$ and $k\neq 0$ is known to be an $N = 1$ superconformal vertex algebra by itself, since it has the Kac-Todorov superconformal vector (see \cites{Kac98, KT85}). Meanwhile, when $\mathfrak{g}$ does not have a non-degenerate even supersymmetric bilinear form or the level $k = 0$, the Kac-Todorov construction collapses and no superconformal vector has been discovered.

In Section \ref{section6}, we begin with the characterization of the $N = n$ SUSY vertex algebras using $n$ odd derivations, as metioned in Section 5.9 of \cite{Kac98}. With this characterization, we observe an action of the orthogonal group on the collection of $N = n$ SUSY structures of the given vertex algebra, which can be viewed as a generalization of the relationships between SUSY structures derived from the superconformal structures (see Proposition \ref{proposition: orthogonal group action}). In particular, we deal with some examples of the SUSY structures associated with superconformal structures.

In Theorem \ref{theorem:N=2 extension}, we generalize the result in Section \ref{section5} to obtain an $N = 2$ SUSY extension of a universal enveloping $N = 1$ SUSY vertex algebra. Moreover, in Remark \ref{remark:N=3 extension}, we show the same argument can be applied to find an $N = n + 1$ SUSY extension of a universal enveloping $N = n$ SUSY vertex algebra. Inductively, a universal enveloping vertex algebra can be extended to an $N = n$ SUSY vertex algebra for $n \geq 1$. In particular, applying the above process to centerless current Lie conformal algebra yields the super-loop algebra in \cite{CFRS90}.

Furthermore, in Proposition \ref{proposition:N=2 superconformal extn of affine}, we show the $N = 2$ SUSY affine vertex algebra ${V}^0_{N=2}(\mathfrak{g})$ of level $0$, which is an $N = 2$ SUSY extension of the affine vertex algebra of level $0$, can be embedded into an $N = 2$ superconformal vertex algebra. More precisely, there is an $N = 2$ superconformal vertex algebra $\widetilde{V}^0_{N=2}(\mathfrak{g})$ such that
\[
\widetilde{V}^0_{N=2}(\mathfrak{g}) \simeq {\textup{SVir}}^{c}_{N=2} \otimes {V}^0_{N=2}(\mathfrak{g}),
\]
where $\textup{SVir}_{N=2}^{c}$ is the $N = 2$  super-Virasoro vertex algebra, and the two odd derivations giving $N=2$ SUSY structure of ${V}^0_{N=2}(\mathfrak{g})$ can be realized by the $N=2$ superconformal structure of $\widetilde{V}^0_{N=2}(\mathfrak{g})$.

\vskip 6mm

\section{Vertex algebras}\label{section2}

In this section, we review basic notions and properties of vertex algebras used throughout this article. For more details, we refer to \cites{DSK06, Kac98}.

\vskip 1mm

In this paper, the base field is $\mathbb{C}$. A vector superspace $V=V_{\bar{0}} \oplus V_{\bar{1}}$ is a $\mathbb{Z}/2\mathbb{Z}$-graded vector space. An element $a\in V$ is called even (resp. odd) if $a\in V_{\bar{0}}$ (resp.  $a\in V_{\bar{1}}$). When $a\in V$ is even (resp. odd), the parity of $a$ is given by $p(a)=0$ (resp. $p(a)=1$).
A vector superspace $V$ is called a {\it vertex algebra} if it is 
 endowed with $m$-th products $a_{(m)}b$ for $a,b\in V$ and $m \in \mathbb{Z}$ with the properties described below.
The $(-1)$-th product $a_{(-1)}b$ is usually denoted by $:ab:$ and called the {\it normally ordered product} of $a$ and $b$. In addition, $V$ has an even derivation $\partial$ such that 
\[ {a}_{(-m-1)}{b} = {\frac{1}{m!}}:({\partial}^{m}{a}){b}: \]
for $m \in \mathbb{N}.$ Therefore the structure of $V$ is determined by (i) the $m$-th product for $m \in \mathbb{Z}_{+}$, (ii) the normally ordered product $: \ :$, and (iii) the derivation $\partial$. In order to deal with all the information from (i) together, we define a bracket $V \otimes V \to \mathbb{C}[\lambda] \otimes V$  by 
\begin{equation} \label{eq:lambda bracket}
(a,b) \mapsto [a \ _\lambda \ b]= \sum_{m \in {\mathbb{Z}}_{+}} {\frac{{\lambda}^{m}}{m!}}{a}_{(m)}{b},
\end{equation}
where $\lambda$ is an even indeterminate. The linear bracket \eqref{eq:lambda bracket} satisfies the following property: 
\begin{enumerate}[]
\item (Sesquilinearity)
\quad $
[\partial a \ _{\lambda} \ b] = - \lambda[a \ _{\lambda} \ b], \quad [a \ _{\lambda} \ \partial b] = (\partial + \lambda)[a \ _{\lambda} \ b].
$
\end{enumerate}
The RHS of the second equality in sesquilinearity can be computed by the relation $\partial \lambda= \lambda \partial$.
More generally, for a $\mathbb{C}[\partial]$-module $R$, a bracket $R \times R \to R[\lambda]$ with the sesquilinearity is called a {\it $\lambda$-bracket} on $R$. In other words, the bracket \eqref{eq:lambda bracket} is a $\lambda$-bracket on the vertex algebra $V$.
Moreover, a vertex algebra is a Lie conformal algebra, which is defined as below.

\begin{definition}[\cite{DSK06}]\label{Definition:LCA}
\rm A \textit{Lie conformal algebra} is a $\mathbb{Z}/2\mathbb{Z}$-graded $\mathbb{C}[\partial]$-module $R$ with a $\lambda$-bracket which is a parity preserving $\mathbb{C}$-linear map
\begin{align*}
[ \ _{\lambda} \ ] : R \otimes R \rightarrow \mathbb{C}[\lambda] \otimes R, \quad a \otimes b \mapsto [a \ _{\lambda} \ b],
\end{align*}
satisfying the following conditions:
\begin{enumerate}[]
\item (Skew-symmetry)
\quad $
[b \ _{\lambda} \ a] = (-1)^{p(a)p(b)+1}[a \ _{-\partial-\lambda} \ b],
$
\item (Jacobi identity) 
\quad 
$
[a \ _{\lambda} \ [b \ _{\gamma} \ c]] = [[a \ _{\lambda} \ b] \ _{{\lambda} + {\gamma}} \ c] + (-1)^{p(a)p(b)}[b \ _{\gamma} \ [a \ _{\lambda} \ c]],
$
\end{enumerate}
where $a,b,c$ are homogeneous elements in  $R$. More precisely, the RHS of skew-symmetry can be rewritten as \[[a \ _{-\partial-\lambda} \ b] =  \sum_{m \in \mathbb{Z}_{+}} \frac{(-\partial-\lambda)^m}{m!} a_{(m)}b,\]
where we denote the $\lambda$-bracket as in \eqref{eq:lambda bracket}, 
and the Jacobi identity is defined in $\mathbb{C}[\lambda,\gamma] \otimes R$. In addition, we assume $\lambda \gamma=\gamma \lambda$.
\end{definition}

\vskip 1mm

Furthermore, there are more properties required to the normally ordered product  on a vertex algebra $V$ and for the compatibility between the normally ordered product and the $\lambda$-bracket or derivation. In the following definition, one can find all axioms for vertex algebras.

\begin{definition}[\cite{DSK06}]\label{definition:Vertex algebra}
\rm
A \textit{vertex algebra} is a tuple $(V, \partial, [ \ _{\lambda} \ ], \ket{0}, : \ :)$ such that:
\\
$\bullet$ $(V, \partial, [ \ _{\lambda} \ ])$ is a Lie conformal algebra,
\\
$\bullet$ $(V, \partial, \ket{0},: \ :)$ is a unital differential superalgebra with a derivation $\partial$, satisfying the following properties:
\begin{enumerate}[]
\item (Quasi-commutativity)
 \quad $\displaystyle :ab: - (-1)^{p(a)p(b)}:ba: = \int_{-\partial}^{0} [a \ _{\lambda} \ b] \ d\lambda,$
\vskip 1mm
\item (Quasi-associativity) \ \ 
\[
::ab:c: - :a:bc:: = :\left( \int_{0}^{\partial} \ d\lambda \  a \right)[b \ _{\lambda} \ c]: + (-1)^{p(a)p(b)}:\left( \int_{0}^{\partial} \ d\lambda \ b \right)[a \ _{\lambda} \ c]:,
\]
\end{enumerate}
\vskip 1mm
$\bullet$ the $\lambda$-bracket and the product $: \ :$ are related by 
\begin{enumerate}[]
\item (Non-commutative Wick formula) 
\[
[a \ _{\lambda} \ :bc:] = :[a \ _{\lambda} \ b]c: + (-1)^{p(a)p(b)}:b[a \ _{\lambda} \ c]: + \int_{0}^{\lambda} [[a \ _{\lambda} \ b] \ _{\gamma} \ c] \ d\gamma.
\]
\end{enumerate}
In the quasi-associativity, we note that $:\big( \int_{0}^{\partial} \ d\lambda \ a \big)[b \ _{\lambda} \ c]: = \sum_{n \in \mathbb{Z}_{+}} \sum_{m \in \mathbb{Z}_{+}}(a_{(-n-2)} (b_{(m)}c)).$ 
\end{definition}

\vskip 1mm

\begin{remark}\label{rem:NO}
\rm
From now on, we denote the normally ordered product $:ab:$ simply by $ab$ for $a,b\in V$. For $a_1,a_2, \cdots, a_l \in V$, the element $a_1 a_2  \cdots a_{l-1} a_l$ is the element obtained by performing normally ordered product from right to left. For example:
    \[ a_1 a_2 a_3 = (a_1 (a_2 a_3)).\]
\end{remark}

\vskip 1mm

As we introduced in Definition \ref{definition:Vertex algebra}, a vertex algebra is a Lie conformal algebra. On the other hand, once we have a Lie conformal algebra, one can extend it to a vertex algebra by the following proposition.

\begin{proposition}[\cite{DSK05}]\label{proposition:universal VA}
Let $R$ be a Lie conformal algebra and consider the Lie bracket on $R$ defined by
\[ [a,b]= \int_{-\partial}^0 d\lambda [a \ _\lambda \ b].\]
Then the universal enveloping algebra $U(R)$ induces a canonical vertex algebra $V(R)$.
More precisely, the associative product on $U(R)$ induces the normally ordered product on $V(R)$ and the unique $\lambda$-bracket on $V(R)$ whose restriction to $R \otimes R$ is the $\lambda$-bracket on $R$ can be obtained by the non-commutative Wick formula.
\end{proposition}

\vskip 1mm

The vertex algebra $V(R)$ in Proposition \ref{proposition:universal VA} is called the {\it universal enveloping vertex algebra} of $R$. By the construction of $V(R)$, one can see that $V(R)$ is strongly generated by $R$. Here, the phrase {\it strongly generated}  means that 
$V(R)$ is spanned by the normally ordered products of a $\mathbb{C}[\partial]$-module $R$.
Moreover, 
since $V(R)$ is defined via the universal enveloping algebra of $R$, one can realize  $V(R)$ is {\it freely generated} by $R$, that is, if $\{v_i |i\in I\}$ is an ordered $\mathbb{C}$-basis of $R$ then 
\[
\{ v_{i_1}v_{i_2}\cdots v_{i_l} \, | \, i_1\leq i_2 \leq \cdots \leq i_l, \ i_t< i_{t+1} \text{ if } p(v_t)=1\}
\]
is a basis of $V(R)$.

\vskip 1mm

\begin{example}[Affine vertex algebra]\label{example: affine VA}
\rm
Let ${\mathfrak{g}}$ be a Lie superalgebra with an even invariant supersymmetric bilinear form $(\, |\, )$. The $\mathbb{C}[\partial]$-module ${\textup{Cur}\mathfrak{g}} = \mathbb{C}[\partial] \otimes \mathfrak{g} \oplus \mathbb{C}K$ with the $\lambda$-brackets
\begin{align*}
[a \ _{\lambda} \ b] = [a, b] + K{\lambda}(a|b),
{\quad}
[K \ _\lambda \ a] = [K \ _\lambda \ K] = 0
\end{align*}
for $a, b \in \mathfrak{g}$, is called the {\it current Lie conformal algebra}.
For $k \in \mathbb{C}$, the quotient of the universal enveloping vertex algebra $V^k(\mathfrak{g}):= V({\textup{Cur}\mathfrak{g}}) \slash \left< K-k\right>$ is called the {\it affine vertex algebra of level $k$}.
\end{example}

\vskip 1mm

\begin{example}[${\beta}{\gamma}$ system and $bc$–${\beta}{\gamma}$ system]\label{ex:def of bcbetagamma}
\rm
Let ${R} = {\mathbb{C}}[{\partial}] \otimes {\textup{span}}_{\mathbb{C}}\{ \beta, \gamma \} \oplus \mathbb{C}C$ be an even Lie conformal algebra with the $\lambda$-brackets
\begin{align*}
[{\beta} \ _{\lambda} \ {\gamma}] = - [{\gamma} \ _{\lambda} \ {\beta}] = C,
\end{align*}
where $C$ is a central element and satisfies ${\partial}C = 0$. Then the universal enveloping vertex algebra $V(R)$ or the quotient vertex algebra $V(R) / \left<C - 1 \right>$ is called the {\it $\beta\gamma$ system}. Let ${\widetilde{R}} = {\mathbb{C}}[{\partial}] \otimes {\textup{span}}_{\mathbb{C}}\{ \beta, \gamma, b, c \} \oplus \mathbb{C}C$ be a Lie conformal algebra, where $b$ and $c$ are odd elements. If the non-zero $\lambda$-brackets between the generators of $\widetilde{R}$ are given by
\begin{align*}
[{\beta} \ _{\lambda} \ {\gamma}] = - [{\gamma} \ _{\lambda} \ {\beta}] = [b \ _{\lambda} \ c] = [c \ _{\lambda} \ b] = C,
\end{align*}
then the universal enveloping vertex algebra $V(\widetilde{R})$ or its quotient $V(\widetilde{R}) / \left<C - 1 \right>$ is called the {\it $bc$–$\beta\gamma$ system}.
\end{example}

\vskip 1mm

\begin{example}[Virasoro vertex algebra]\label{example: Virasoro LCA}
\rm
The Lie conformal algebra ${\textup{Vir}} = {\mathbb{C}}[\partial] \otimes {\mathbb{C}}L \oplus {\mathbb{C}}C$ endowed with the $\lambda$-brackets
\begin{equation} \label{eq:virasoro}
[L \ _\lambda \ L] = ({\partial} + 2{\lambda})L + {\frac{C}{12}}{\lambda}^{3},
{\quad}
[C \ _\lambda \ L] = [C \ _\lambda \ C] = 0
\end{equation}
is called the {\it Virasoro Lie conformal algebra}. For $c \in \mathbb{C}$, the {\it Virasoro vertex algebra of central charge $c$} is the quotient vertex algebra ${\textup{Vir}}^{c}:= V({\textup{Vir}}) \slash \left< C-c\right>$.
\end{example}

\vskip 1mm

In a vertex algebra $V$, an even vector $L \in V$ is called a \textit{conformal vector} if $L_{(0)} = \partial$, $L_{(1)}$ is diagonalizable on $V$ and it satisfies the Virasoro relation \eqref{eq:virasoro} with central charge $c$. If a vertex algebra has a conformal vector then it is called a \textit{conformal vertex algebra}. For a conformal vector $L$ of a conformal vertex algebra $V$, if $a\in V$ satisfies $[L\ {}_\lambda \ a]=(\partial+ \Delta_a)a + O(\lambda^2)$ for $\Delta_a \in \mathbb{C}$, then $\Delta_a$ is called the {\it conformal weight} of $a$. If $O(\lambda^2)=0$, then we call $a$ a  {\it primary of conformal weight $\Delta_a$.}

\vskip 1mm

\begin{example}[\cites{HK07, Kac98, KW04}]\label{example:superVirasoro VA}
\rm
The following vertex algebras are the most well-known examples of the extensions of Virasoro vertex algebras.
The \textit{$N = 1$ super-Virasoro vertex algebra} ${\textup{SVir}}^{c}$ is generated by an even vector $L$ and an odd vector $G$ satisfying the following super-Virasoro relation
\begin{align} \label{eq:superVirasoro relation}
[L \ _\lambda \ L] = ({\partial} + 2{\lambda})L + {\frac{c}{12}}{\lambda}^{3},
{\quad}
[L \ _\lambda \ G] = ({\partial} + {\frac{3}{2}}{\lambda})G,
{\quad}
[G \ _\lambda \ G] = 2L + {\frac{c}{3}}{\lambda}^{2}.
\end{align}
The \textit{$N = 2$ super-Virasoro vertex algebra} ${\textup{SVir}}^{c}_{N=2}$ is generated by two even vectors $L$, $J$, and two odd vectors $G^{+}$, $G^{-}$, satisfying the following $N = 2$ super-Virasoro relation
\begin{equation}\label{eq:N=2 superVirasoro relation}
\begin{gathered}\relax
[L \ _\lambda \ L] = ({\partial} + 2{\lambda})L + {\frac{c}{12}}{\lambda}^{3},
{\quad}
[L \ _\lambda \ G^{\pm}] = (\partial + \frac{3}{2}\lambda){G}^{\pm},
{\quad}
[G^{\pm} \ _\lambda \ G^{\pm}] = 0,
\\
[G^{+} \ _\lambda \ G^{-}] = L + (\frac{1}{2}\partial + \lambda)J + \frac{c}{6}{\lambda}^{2},
{\quad}
[L \ _\lambda \ J] = (\partial + \lambda)J,
{\quad}
[G^{\pm} \ _\lambda \ J] = {\mp}{G}^{\pm},
{\quad}
[J \ _\lambda \ J] = \frac{c}{3}\lambda.
\end{gathered}
\end{equation}
The \textit{$N = 3$ super-Virasoro vertex algebra} is generated by a conformal vector $L$ with central charge $c$, three even primary vectors ${J}^{+}$, ${J}^{-}$, ${J}^{0}$ of conformal weight $1$, three odd primary vectors ${G}^{+}$, ${G}^{-}$, ${G}^{0}$ of conformal weight $\frac{3}{2}$, and an odd primary vector ${\Phi}$ of conformal weight $\frac{1}{2}$, satisfying the following $N = 3$ super-Virasoro relation where the non-zero ${\lambda}$-brackets are:
\begin{equation}\label{eq:N=3 superVirasoro relation}
\begin{gathered}\relax
[{G}^{+} \ _\lambda \ {G}^{-}] = {L} + ({\frac{1}{2}}{\partial} + {\lambda}){J}^{0} + {\frac{c}{6}}{\lambda}^{2},
{\quad}
[{G}^{\pm} \ _\lambda \ {G}^{0}] = \pm({\frac{1}{2}}{\partial} + {\lambda}){J}^{\pm},
{\quad}
[{G}^{0} \ _\lambda \ {G}^{0}] = L + {\frac{c}{6}}{\lambda}^{2},
\\
[{J}^{+} \ _\lambda \ {J}^{-}] = {J}^{0} + {\frac{c}{3}}{\lambda},
{\quad}
[{J}^{\pm} \ _\lambda \ {J}^{0}] = {\mp}{J}^{\pm},
{\quad}
[{J}^{0} \ _\lambda \ {J}^{0}] = {\frac{c}{3}}{\lambda},
\\
[{G}^{\pm} \ _\lambda \ {J}^{\mp}] = \mp{G}^{0} + ({\partial} + {\lambda}){\Phi},
{\quad}
[{G}^{\pm} \ _\lambda \ {J}^{0}] = [{G}^{0} \ _\lambda \ {J}^{\pm}] = {\mp}{G}^{\pm},
{\quad}
[{G}^{0} \ _\lambda \ {J}^{0}] = -({\partial} + {\lambda}){\Phi},
{\quad}
\\
[{G}^{\pm} \ _\lambda \ {\Phi}] = {\frac{1}{2}}{J}^{\pm},
{\quad}
[{G}^{0} \ _\lambda \ {\Phi}] = -{\frac{1}{2}}{J}^{0},
{\quad}
[{\Phi} \ _\lambda \ {\Phi}] = {\frac{c}{6}}.
\end{gathered}
\end{equation}
Note that, for the $N = 3$ super-Virasoro vertex algebra, there is an automorphism of the form
\begin{align*}
{G}^{\pm} \mapsto {G}^{\mp},
{\quad}
{G}^{0} \mapsto -{G}^{0},
{\quad}
{J}^{\pm} \mapsto {J}^{\mp},
{\quad}
{J}^{0} \mapsto -{J}^{0}.
\end{align*}
For ${\mu} \in {\mathbb{C}}^{*}$, the following map defined by
${G}^{\pm} \mapsto {\mu}^{\pm}{G}^{\pm}$, ${J}^{\pm} \mapsto {\mu}^{\pm}{J}^{\pm}$ is also an automorphism of the $N = 3$ super-Virasoro vertex algebra. For more details, the reader is referred to \cite{SS87}.
\end{example}

\vskip 6mm

\section{Supersymmetric vertex algebras}\label{section3}

In this section, we introduce supersymmetric (SUSY) vertex algebras 
using so-called $\Lambda$-bracket. It can be understood as an analogous definition to that of vertex algebras in Definition \ref{definition:Vertex algebra}. We remark, all the SUSY vertex algebras considered in this paper are the $N_{K} = n$ SUSY vertex algebras in \cite{HK07}. For the simplicity of notations, let us denote by $N = n$ SUSY vertex algebras, instead of $N_{K} = n$ SUSY vertex algebras. The main reference of this section is \cite{HK07}.

\vskip 1mm

Let $R$ be a $\mathbb{C}[\nabla]$-module, where $\nabla=(\partial, D^1, D^2, \cdots, D^n)$ consists of an even operator $\partial$ and   odd operators $D^i$   such that 
\begin{equation}\label{eq:Der commutator}
 [D^i, D^j] = \delta_{i,j} 2{\partial}
\end{equation}
for $i,j=1,2,\cdots, n$.
Note that  the bracket in \eqref{eq:Der commutator} is a supercommutator. 
In addition, consider a tuple $\Lambda:= (\lambda, \chi^1, \chi^2, \cdots, \chi^n)$ of an even formal variable $\lambda$ and odd formal variables $\chi^i$ which are subject to the relations 
\begin{equation*} \label{eq:var commutator}
[\chi^i, \chi^j]= -\delta_{i,j} 2 \lambda,
\end{equation*}
where $i,j=1,2, \cdots, n$, and the bracket is again a supercommutator. Now, a {\it $\Lambda$-bracket} $R \otimes R \to \mathbb{C}[\Lambda] \otimes R$ is a degree $\bar{n}$ linear map with the following property:
\begin{enumerate}[]
\item (Sesquilinearity)
\quad 
$[{D}^{i}a \ _{\Lambda} \ b] = (-1)^{n+1}{\chi}^{i}[a \ _{\Lambda} \ b], \quad [a \ _{\Lambda} \ {D}^{i}b] = (-1)^{p(a)+n}({D}^{i} + {\chi}^{i})[a \ _{\Lambda} \ b]$.
\end{enumerate}
The RHS of the second equality in sesquilinearity can be computed by the relation
\begin{equation}\label{eq:D,chi}
 [D^i, \chi^j]= \delta_{i,j} 2 \lambda, \quad [D^i, \lambda]=0.
\end{equation}

\begin{definition}[\cites{HK07}]\label{definition:SUSY LCA}
\rm An $N = n$ \textit{supersymmetric {\textup{(}}SUSY{\textup{)}} Lie conformal algebra} is a $\mathbb{Z}/2\mathbb{Z}$-graded $\mathbb{C}[\nabla]$-module $R$ with a $\Lambda$-bracket which is a $\mathbb{C}$-linear map of degree $\bar{n}$:
\begin{align*}
[ \ _{\Lambda} \ ] : R \otimes R \rightarrow \mathbb{C}[\Lambda] \otimes R, \quad a \otimes b \mapsto [a \ _{\Lambda} \ b],
\end{align*}
satisfying the following conditions:
\begin{enumerate}[]
\item (Skew-symmetry)
\quad $
[b \ _{\Lambda} \ a] = (-1)^{p(a)p(b)+n+1}[a \ _{-\nabla-\Lambda} \ b]$,
\item (Jacobi identity)
\quad $
[a \ _{\Lambda} \ [b \ _{\Gamma} \ c]] = (-1)^{(p(a)+1)n}[[a \ _{\Lambda} \ b] \ _{\Lambda+\Gamma} \ c] + (-1)^{(p(a)+n)(p(b)+n)}[b \ _{\Gamma} \ [a \ _{\Lambda} \ c]]$,
\end{enumerate}
for $a, b, c \in R$. We remark that the Jacobi identity holds in $\mathbb{C}[\Lambda, \Gamma]$, where $\Gamma = (\gamma, {\eta}^{1}, {\eta}^{2}, \cdots, {\eta}^{n})$ consists of an even formal variable ${\gamma}$ and odd formal variables ${\eta}^i$ with the relations $[{\eta}^{i}, {\eta}^{j}] = -{\delta}_{i,j}2{\gamma}$ and $\Lambda$ and $\Gamma$ are supercommute. In addition, $\mathbb{C}[\Lambda,\Gamma] \otimes R$ is a $\mathbb{C}[\nabla]$-module via the relation \eqref{eq:D,chi} and 
\begin{align*}
[D^i, \eta^j]= \delta_{i,j} 2 \gamma, \quad [D^i, \gamma]=0.
\end{align*}
In addition, the $\Lambda$-bracket supercommutes with $\mathbb{C}[\Lambda, \Gamma]$. For example, $[ a\ {}_\Lambda \  \eta^i b] = (-1)^{p(a)+n} \eta^i [a\ {}_\Lambda \ b]$ for any $a,b\in R$.
\end{definition}

\vskip 1mm

In particular, we denote the $\Lambda$-bracket for $N = 1$ SUSY Lie conformal algebras by
\begin{align*}
[a \ _{\Lambda} \ b] = \sum\limits_{m \in \mathbb{Z}_{+}}\frac{{\lambda}^{m}}{m!}a_{(m|0)}b + \chi\sum\limits_{m \in \mathbb{Z}_{+}}\frac{{\lambda}^{m}}{m!}a_{(m|1)}b,
\end{align*}
and the $\Lambda$-bracket for $N = 2$ SUSY Lie conformal algebras by 
\begin{align*}
[a \ _{\Lambda} \ b]
 = \sum\limits_{m \in \mathbb{Z}_{+}}\frac{\lambda^{m}}{m!}{a}_{(m|00)}b
 - {\chi}^{1}\sum\limits_{m \in \mathbb{Z}_{+}}\frac{\lambda^{m}}{m!}{a}_{(m|10)}b - {\chi}^{2}\sum\limits_{m \in \mathbb{Z}_{+}}\frac{\lambda^{m}}{m!}{a}_{(m|01)}b
 - {\chi}^{1}{\chi}^{2}\sum\limits_{m \in \mathbb{Z}_{+}}\frac{\lambda^{m}}{m!}{a}_{(m|11)}b.
\end{align*}

\vskip 1mm

\begin{definition}[\cite{HK07}]\label{definition:SUSY VA}
\rm An $N = n$ \textit{supersymmetric {\textup{(}}SUSY{\textup{)}} vertex algebra} is a tuple $(V, \nabla, [ \ _{\Lambda} \ ], \ket{0}, : \ :)$ such that:
\\
$\bullet$ $(V, \nabla , [ \ _{\Lambda} \ ])$ is an $N = n$ SUSY Lie conformal algebra,
\\
$\bullet$ $(V, \nabla, \ket{0},: \ :)$ is a unital differential superalgebra satisfying the following properties:
\begin{enumerate}[]
\item (Quasi-commutativity)
 \quad $\displaystyle :ab: - (-1)^{p(a)p(b)}:ba: = \int_{-\nabla}^{0} [a \ _{\Lambda} \ b] \ d\Lambda ,$
\vskip 1mm
\item (Quasi-associativity) \ \ 
\[
::ab:c: - :a:bc:: = :\left( \int_{0}^{\nabla} \ d\Lambda a \right)[b \ _{\Lambda} \ c]: + (-1)^{p(a)p(b)} :\left( \int_{0}^{\nabla} \ d\Lambda b \right)[a \ _{\Lambda} \ c]:,
\]
\end{enumerate}
\vskip 1mm
$\bullet$ the $\Lambda$-bracket and the product $: \ :$ are related by
\begin{enumerate}[]
\item (Non-commutative Wick formula) 
\[
[a \ _{\Lambda} \ :bc:] = :[a \ _{\Lambda} \ b]c : + (-1)^{(p(a)+n)p(b)}:b[a \ _{\Lambda} \ c]: + \int_{0}^{\Lambda} [[a \ _{\Lambda} \ b] \ _{\Gamma} \ c] \ d\Gamma,
\]
\end{enumerate}
where the integral $\int_{0}^{\Lambda} d\Gamma$ is computed as ${\partial}_{{\eta}^{1}}{\partial}_{{\eta}^{2}} \cdots {\partial}_{{\eta}^{n}}\int_{0}^{\lambda} d\gamma$.
\end{definition}

\vskip 1mm

As an analogue of the universal enveloping vertex algebras, one can construct the universal enveloping $N = n$ SUSY vertex algebra $V(R)$ of an $N = n$ SUSY Lie conformal algebra $R$. More precisely, $V(R)$ is freely generated by $R$ with respect to the $N = n$ SUSY normally ordered product and the $\Lambda$-bracket on $V(R)$ is induced from the $\Lambda$-bracket on $R$ and the $N = n$ SUSY Wick formula (see \cite{HK07}). We also note that, throughout this article, the normally ordered product $:ab:$ is denoted by $ab$ for $a, b \in V$ as in the non-SUSY case (see Remark \ref{rem:NO}) to simplify notations.

\vskip 1mm

\begin{example}[\cites{HK07}]\label{example:SUSY affine VA}
\rm
Let ${\mathfrak{g}}$ be a Lie superalgebra with an even invariant supersymmetric bilinear form $(\, |\, )$. For an even $n$, the $\mathbb{C}[{\nabla}]$-module $R = \mathbb{C}[{\nabla}] \otimes \mathfrak{g} \oplus \mathbb{C}K$ with the $\Lambda$-brackets
\begin{align*}
[a \ _{\Lambda} \ b] = [a, b] + K{\lambda}(a|b)
\end{align*}
for $a, b \in \mathfrak{g}$, is an $N = n$ SUSY Lie conformal algebra.
For an odd $n$, the $\mathbb{C}[{\nabla}]$-module $R = \mathbb{C}[{\nabla}] \otimes \bar{\mathfrak{g}} \oplus \mathbb{C}K$ with the $\Lambda$-brackets
\begin{align*}
[{\bar{a}} \ _{\Lambda} \ {\bar{b}}] = (-1)^{p(a)}\big(\overline{[a, b]} + K(\sum_{i = 1}^{n} {\chi}^{i})(a|b)\big)
\end{align*}
for $\bar{a}, \bar{b} \in \bar{\mathfrak{g}}$, is an $N = n$ SUSY Lie conformal algebra. Here $\bar{\mathfrak{g}}$ denotes ${\mathfrak{g}}$ with reversed parity. The quotient of the universal enveloping $N=n$ SUSY vertex algebra $V_{N=n}^k(\mathfrak{g}):= V(R)/\left< K-k\right>$ is called the {\it $N = n$ SUSY affine vertex algebra of level $k$}.
\end{example}

\vskip 1mm

We remark that every $N = n$ SUSY vertex algebra $V$ is an $N = n'$ SUSY vertex algebra for $n' < n$. In particular, if $V$ is an  $N=1$ SUSY vertex algebra, the $m$-th product 
$a_{(m)}b:= a_{(m|1)}b$ for $a,b\in V$ and $m \in \mathbb{Z}_+$ induces a vertex algebra structure on $V$ together with the normally ordered product and the derivation $\partial=\frac{1}{2}[D,D]$. For instance, the $N = 1$ SUSY affine vertex algebra ${V}^{k}_{N=1}(\mathfrak{g})$ in Example \ref{example:SUSY affine VA} is a vertex algebra strongly generated by $\mathbb{C}[\partial] \otimes (\bar{\mathfrak{g}} \oplus D{\bar{\mathfrak{g}}})$. In addition, the differential algebra homomorphism $V^k(\mathfrak{g}) \to {V}^{k}_{N=1}(\mathfrak{g})$ for the affine vertex algebra $V^k(\mathfrak{g})$ (see Example \ref{example: affine VA}) defined by $a \mapsto D\bar{a}$ is an injective vertex algebra homomorphism. In other words, $V^k(\mathfrak{g})$ is a proper vertex subalgebra of ${V}^{k}_{N=1}(\mathfrak{g})$. If $V$ is an $N=2$ SUSY vertex algebra, the $m$-th product $a_{(m)}b:= a_{(m|11)}b$ for $m\in \mathbb{Z}_+$ induces a vertex algebra structure on $V$.

\vskip 1mm

\begin{example}[\cites{HK07}]\label{example: SUSY bcbetagamma}
\rm
Recall that the $bc$–$\beta\gamma$ system in Example \ref{ex:def of bcbetagamma} is generated by two even fields $\beta$, $\gamma$ and two odd fields $b$, $c$ as a vertex algebra. On the other hand, if we define an odd derivation $D$ on the $bc$–$\beta\gamma$ system by
\begin{align*}
D(b) = {\beta},
{\quad}
D(c) = {\partial}{\gamma},
{\quad}
D({\beta}) = {\partial}{b},
{\quad}
D({\gamma}) = c,
\end{align*}
then the $bc$–$\beta\gamma$ system can be regarded as an $N = 1$ SUSY vertex algebra generated by one even field $\gamma$ and one odd field $b$ with the following $N = 1$ $\Lambda$-bracket relations:
\begin{align*}
[b \ _{\Lambda} \ {\gamma}] = [{\gamma} \ _{\Lambda} \ b] = C.
\end{align*}
\end{example}

\vskip 1mm

\begin{example}[\cites{EHZ12, HK07}] \label{example:SUSY superVirasoro VA}
\rm
As an $N = 1$ SUSY vertex algebra, the $N = 1$ super-Virasoro vertex algebra in Example \ref{example:superVirasoro VA} is generated by an odd vector $G$ with the $N = 1$ $\Lambda$-bracket relation (see Example 5.5 of \cite{HK07})
\begin{align*}
[G \ _\Lambda \ G] = (2{\partial} + 3{\lambda} + {\chi}{D})G + {\frac{c}{3}}{\lambda}^{2}{\chi},
\end{align*}
and the $N = 2$ super-Virasoro vertex algebra is generated by an even vector $J$ as an $N = 2$ SUSY vertex algebra (see Section 3.2 of \cite{EHZ12}), with the $N = 2$ $\Lambda$-bracket relation
\begin{align*}
[J \ _\Lambda \ J] = (2{\partial} + 2{\lambda} + {\chi}^{1}{D}^{1} + {\chi}^{2}{D}^{2})J + {\frac{c}{3}}{\lambda}{\chi}^{1}{\chi}^{2}.
\end{align*}
\end{example}

\vskip 6mm

\section{Various superconformal structures of a vertex algebra}\label{section:sconf str of VA}
\setcounter{equation}{0}

In a vertex algebra $V$, a superconformal vector induces the conformal weight decomposition of $V$. The structure arising from the decomposition is called the superconformal structure of $V$. It is closely related to a SUSY structure in the sense that both of them can be derived from a superconformal vector. In this section, we describe various superconformal structures correlated with the same $N = 1$ SUSY structure of a vertex algebra by performing so-called `shifts' of a given superconformal vector.

Recall that an odd vector $G$ of a vertex algebra $V$ is called a {\it superconformal vector} if it satisfies the super-Virasoro relation \eqref{eq:superVirasoro relation} and $L = {\frac{1}{2}}{G}_{(0)}{G}$ is a conformal vector. In this case, we say that $V$ is an {\it $N = 1$ superconformal vertex algebra}. In the context of $N = 1$ SUSY vertex algebras, as in Section 1.2 of \cite{HK07}, an $N = 1$ superconformal vertex algebra $V$ with a superconformal vector $G \in V$ has an $N = 1$ SUSY vertex algebra structure by taking the odd derivation $D := G_{(0)}$. The following examples are superconformal vectors of $N = 1$ SUSY affine vertex algebras and the $bc$–${\beta}{\gamma}$ system.

\begin{example}[$N = 1$ superconformal structure of $N = 1$ SUSY affine vertex algebra]
\rm
Let ${\mathfrak{g}}$ be a basic Lie superalgebra with a non-degenerate invariant even supersymmetric bilinear form $(\, |\, )$, and let ${V}^{k}_{N=1}(\mathfrak{g})$ be the $N = 1$ SUSY affine vertex algebra of level $k \neq 0$ in Example \ref{example:SUSY affine VA}. Then, from Example 5.9 of \cite{HK07}, the following vector
\begin{align}\label{eq:KT sconf}
{\tau}^{\mathfrak{g}} = {\frac{1}{k}}\Big(\sum_{i \in I}(-1)^{p({a}^{i})}(D{\bar{a}}^{i}){\bar{b}}^{i} + {\frac{1}{3k}}\sum_{i, j, r \in I}([{a}^{i}, {a}^{j}] \, | \, {a}^{r}){\bar{b}}^{i}{\bar{b}}^{j}{\bar{b}}^{r}\Big),
\end{align}
is a superconformal vector of ${V}^{k}_{N=1}(\mathfrak{g})$, where $\{ {a}^{i} \}_{ i {\in} I }$ and $\{ {b}^{i} \}_{ i {\in} I }$ are dual bases of $\mathfrak{g}$ with respect to $(\, |\, )$. The superconformal vector ${\tau}^{\mathfrak{g}}$ is called the \textit{Kac-Todorov superconformal vector} (see also \cites{Kac98, KT85}), and the central charge ${c}^{\mathfrak{g}}$ is given by
\begin{align*}
{c}^{\mathfrak{g}} = \frac{(k - {h}^{\vee}){\textup{sdim}}{\mathfrak{g}}}{k} + \frac{{\textup{sdim}}{\mathfrak{g}}}{2},
\end{align*}
where ${h}^{\vee}$ is the dual Coxeter number of ${\mathfrak{g}}$. Note that the construction of the superconformal vector ${\tau}^{\mathfrak{g}}$ can be thought of as a combination of a supersymmetrization of the Sugawara construction and the standard superconformal vector of the free fermions (see also \cites{DSK06, IR94}). For the $N = 2$ superconformal structures of the $N = 1$ SUSY affine vertex algebra of level $k\neq 0$ associated with the double of a Lie bialgebra, the reader is referred to \cites{HZ10, Parkhomenko92}.
\end{example}

\begin{example}[$N = 1$ superconformal structure of $bc$–${\beta}{\gamma}$ system]\label{N=1 sconf of bcbetagamma}
\rm
For the $bc$–${\beta}{\gamma}$ system in Example \ref{ex:def of bcbetagamma} (see also Example \ref{example: SUSY bcbetagamma}), the odd vector $({\partial}{\gamma}){b} + {c}{\beta}$ is a superconformal vector of the $bc$–${\beta}{\gamma}$ system, where the corresponding conformal vector is $\frac{1}{2}(-{c}{\partial}{b} + ({\partial}{c}){b} + 2({\partial}{\gamma}){\beta})$ with central charge $3$. Moreover, one can check that the superconformal vector $({\partial}{\gamma}){b} + {c}{\beta}$ induces the odd derivation $D$ in Example \ref{example: SUSY bcbetagamma}, that is, $D = (({\partial}{\gamma}){b} + {c}{\beta})_{(0)}$.
\end{example}

\vskip 1mm

In general, for an $N = 1$ superconformal vertex algebra, if there is an odd vector satisfying the conditions in the following proposition, we can shift the superconformal vector to another superconformal vector, giving rise to various conformal weights. At the level of vertex algebra structures, the {\textit{shifting}} of a superconformal vector is a superconformal vector that yields the same $N = 1$ SUSY structures but a different superconformal structure.

\begin{proposition}\label{proposition:shifted sconf}
Let $V$ be an $N = 1$ superconformal vertex algebra with a superconformal vector $G$, where $L$ is the conformal vector with central charge $c$. Suppose that $v \in V$ is an odd vector satisfying the following relations
\begin{align}\label{eq:shift of sconf}
[{L} \ _{\lambda} \ {v}] = ({\partial} + \frac{\lambda}{2}){v},
{\quad}
[{G} \ _{\lambda} \ {v}] = {G}_{(0)}{v} + {c}_{1}{\lambda},
{\quad}
[{v} \ _{\lambda} \ {v}] = {c}_{2},
\end{align}
for some ${c}_{1}, {c}_{2} \in \mathbb{C}$ and $({G}_{(0)}{v})_{(0)}$ is diagonalizable on $V$. Then, the vector $G + {\partial}v$ is also a superconformal vector and the corresponding conformal vector is $L + {\frac{1}{2}}{\partial}{G}_{(0)}{v}$ with central charge $c + 6{c}_{1} - 3{c}_{2}$.
\end{proposition}

\begin{proof}
It follows from the $\lambda$-bracket calculations using the properties of Lie conformal algebras. For example, we know that the following $\lambda$-bracket relations hold:
\begin{align*}
[{G}_{(0)}{v} \ _{\lambda} \ {G}] = {\lambda}{v},
{\quad}
[{G}_{(0)}{v} \ _{\lambda} \ {v}] = 0,
\end{align*}
hence, by the Jacobi identity, we obtain:
\begin{align*}
[{G}_{(0)}{v} \ _{\lambda} \ {G}_{(0)}{v}] 
= [{G}_{(0)}{v} \ _{\lambda} \ [{G} \ _{\gamma} \ {v}]] 
= [[{G}_{(0)}{v} \ _{\lambda} \ {G}] \ _{\lambda + \gamma} \ {v}] + [{G} \ _{\gamma} \ [{G}_{(0)}{v} \ _{\lambda} \ {v}]] 
= [{\lambda}{v} \ _{\lambda + \gamma} \ {v}] 
= {\lambda}{c}_{2}.
\end{align*}
On the other hand, from the Jacobi identity, we have:
\begin{align*}
[{L} \ _{\lambda} \ {G}_{(0)}{v}] 
&= [{L} \ _{\lambda} \ [{G} \ _{\gamma} \ {v}]] 
= [[{L} \ _{\lambda} \ {G}] \ _{\lambda + \gamma} \ {v}] + [{G} \ _{\gamma} \ [{L} \ _{\lambda} \ {v}]]
\\
&= [({\partial} + {\frac{3}{2}}{\lambda})G \ _{\lambda + \gamma} \ {v}] + [{G} \ _{\gamma} \ ({\partial} + {\frac{1}{2}}{\lambda}){v}]
\\
&= ({\frac{1}{2}}{\lambda} - {\gamma})({G}_{(0)}{v} + {c}_{1}({\lambda} + {\gamma})) + ({\partial} + {\frac{1}{2}}{\lambda} + {\gamma})({G}_{(0)}{v} + {c}_{1}{\gamma})
\\
&= ({\partial} + {\lambda}){G}_{(0)}{v} + {\frac{{c}_{1}}{2}}{\lambda}^{2}.
\end{align*}
Combining the above $\lambda$-bracket relations, we obtain the following $\lambda$-bracket result:
\begin{align*}
[{L} + {\frac{1}{2}}{\partial}{G}_{(0)}{v} \ _{\lambda} \ {L} + {\frac{1}{2}}{\partial}{G}_{(0)}{v}] 
&= [{L} \ _{\lambda} \ {L}] + {\frac{1}{2}}({\partial} + {\lambda})[{L} \ _{\lambda} \ {G}_{(0)}{v}] - {\frac{1}{2}}{\lambda}[{G}_{(0)}{v} \ _{\lambda} \ {L}] - {\frac{1}{4}}{\lambda}({\partial} + {\lambda})[{G}_{(0)}{v} \ _{\lambda} \ {G}_{(0)}{v}]
\\
&= ({\partial} + 2{\lambda}){L} + {\frac{c}{12}}{\lambda}^{3} + {\frac{1}{2}}({\partial} + {\lambda})({\partial} + {\lambda}){G}_{(0)}{v} + {\frac{{c}_{1}}{2}}{\lambda}^{3} - {\frac{1}{2}}{\lambda}^{2}{G}_{(0)}{v} - {\frac{{c}_{2}}{4}}{\lambda}^{3}
\\
&= ({\partial} + 2{\lambda})({L} + {\frac{1}{2}}{\partial}{G}_{(0)}{v}) + {\frac{1}{12}}({c} + 6{c}_{1} - 3{c}_{2}){\lambda}^{3}.
\end{align*}
Similarly, the other $\lambda$-brackets of the super-Virasoro relation \eqref{eq:superVirasoro relation} can be verified.
\end{proof}

\vskip 1mm

Note that, for a vertex algebra $V$ with a superconformal vector $G$, if we consider the $N = 1$ SUSY vertex algebra structure of $V$ via $D = {G}_{(0)}$, then \eqref{eq:shift of sconf} in Proposition \ref{proposition:shifted sconf} can be expressed as follows:
\begin{align}\label{eq:SUSY ver shift of sconf}
[{G} \ _{\Lambda} \ {v}] = (2{\partial} + {\lambda} +{\chi}{D}){v} + {c}_{1}{\lambda}{\chi},
{\quad}
[{v} \ _{\Lambda} \ {v}] = {c}_{2}{\chi},
\end{align}
in terms of $N = 1$ $\Lambda$-brackets.

\vskip 1mm

\begin{example}[Shifted superconformal vector of $N = 1$ SUSY affine vertex algebra]
\rm
For the Kac-Todorov superconformal vector ${\tau}^{\mathfrak{g}}$ in \eqref{eq:KT sconf} and an element $v$ of the Cartan subalgebra of $\mathfrak{g}$, the following odd vector
\begin{align}\label{eq:shifted KT}
{\tau}^{\mathfrak{g}}_{v} = {\tau}^{\mathfrak{g}} + {\partial}{\bar{v}},
\end{align}
is a superconformal vector of the $N = 1$ SUSY affine vertex algebra of level $k \neq 0$ (see also Section 3.4 of \cite{DSK06} and Example 2.13 of \cite{HZ10}), since the conditions in Proposition \ref{proposition:shifted sconf} hold with ${c}_{1} = 0$ and ${c}_{2} = k(v|v)$.
\end{example}

\vskip 1mm

In the rest of this section, let us introduce the SUSY BRST complex $\mathcal{C}(\bar{\mathfrak{g}}, f, k)$ for a basic Lie superalgebra $\mathfrak{g}$ and an odd element $f$ in $\mathfrak{osp}(1|2)$-subalgebra, and a 1-parameter family of superconformal vectors of $\mathcal{C}(\bar{\mathfrak{g}}, f, k)$. Consider the Dynkin grading $\mathfrak{g} = \bigoplus_{i \in \frac{1}{2}{\mathbb{Z}}}{\mathfrak{g}}(i)$ such that $f \in \mathfrak{g}(-\frac{1}{2})$ and denote the subspaces ${\mathfrak{n}} = \bigoplus_{i > 0}{\mathfrak{g}}(i)$ and ${\mathfrak{n}_{-}} = \bigoplus_{i < 0}{\mathfrak{g}}(i)$ of ${\mathfrak{g}}$, and take bases $\{ {u}_{\alpha} \}_{ \alpha {\in} {I}^{+} }$ of ${\mathfrak{n}}$ and $\{ {u}^{\alpha} \}_{ \alpha {\in} {I}^{+} }$ of ${\mathfrak{n}}_{-}$ satisfying $({u}^{\alpha}|{u}_{\beta}) = {\delta}_{\alpha, \beta}$.
The {\it SUSY charged free fermion vertex algebra} ${F}^{\textup{ch}}_{N=1}(\mathfrak{g}, f)$ is the $N = 1$ SUSY vertex algebra freely generated by ${\phi}_{\alpha}$ and ${\phi}^{\bar{\beta}}$ for ${\alpha}, {\beta} \in {I}^{+}$, where ${\phi}_{\mathfrak{n}} \simeq {\mathfrak{n}}$ and ${\phi}^{{\bar{\mathfrak{n}}}_{-}} \simeq {\bar{\mathfrak{n}}}_{-}$ as vector superspaces and 
\begin{align*}
[{\phi}_{\alpha} \ _{\Lambda} \ {\phi}^{\bar{\beta}}] = ({u}^{\beta}|{u}_{\alpha}),
{\quad}
[{\phi}_{\alpha} \ _{\Lambda} \ {\phi}_{\beta}] = [{\phi}^{\bar{\alpha}} \ _{\Lambda} \ {\phi}^{\bar{\beta}}] = 0,
\end{align*}
for ${\alpha}, {\beta} \in {I}^{+}$. Now we obtain the $N = 1$ SUSY vertex algebra called the SUSY BRST complex:
\begin{align*}
\mathcal{C}(\bar{\mathfrak{g}}, f, k) = {V}^{k}_{N=1}(\mathfrak{g}) \otimes {F}^{\textup{ch}}_{N=1}(\mathfrak{g}, f).
\end{align*}
We note that the SUSY BRST complex $\mathcal{C}(\bar{\mathfrak{g}}, f, k)$ is the main ingredient of the SUSY $W$–algebra associated with $\bar{\mathfrak{g}}$ and $f$ (see Section 4.1 of \cite{MRS21}). In order to obtain the superconformal vectors of $\mathcal{C}(\bar{\mathfrak{g}}, f, k)$, those of ${F}^{\textup{ch}}_{N=1}(\mathfrak{g}, f)$ are needed.

\begin{example}[Superconformal vector of SUSY charged free fermion vertex algebra]
\rm
Let ${I}^{+} = {I}^{+}_{\bar{0}} \sqcup {I}^{+}_{\bar{1}}$ for convenience, where ${I}^{+}_{\bar{0}} = \{ {\alpha} \in {I}^{+} \, | \, p({u}_{\alpha}) = 0 \}$ and ${I}^{+}_{\bar{1}} = \{ {\alpha} \in {I}^{+} \, | \, p({u}_{\alpha}) = 1 \}$.
Note that the SUSY charged free fermion vertex algebra ${F}^{\textup{ch}}_{N=1}({\mathfrak{g}}, f)$ is isomorphic to the ${\textup{dim}}{\mathfrak{n}}$ copies of the $bc$–${\beta}{\gamma}$ system. Hence, by using the superconformal vector in Example \ref{N=1 sconf of bcbetagamma}, we know that the following vector
\begin{align}\label{eq:standard sconf of chFF}
{\tau}^{\textup{ch}} = \sum\limits_{{\alpha} \in {I}^{+}_{\bar{0}}}\big( ({\partial}{\phi}_{\alpha}){\phi}^{\bar{\alpha}} + D{\phi}_{\alpha}D{\phi}^{\bar{\alpha}} \big)
 + \sum\limits_{{\alpha} \in {I}^{+}_{\bar{1}}}\big( {\phi}_{\alpha}{\partial}{\phi}^{\bar{\alpha}} + D{\phi}_{\alpha}D{\phi}^{\bar{\alpha}} \big),
\end{align}
is a superconformal vector of ${F}^{\textup{ch}}_{N=1}({\mathfrak{g}}, f)$ with central charge ${c}^{\textup{ch}} = 3{\textup{dim}}{\mathfrak{n}}$ (see also Example 5.12 of \cite{HK07}).
\end{example}

\vskip 1mm

Moreover, we can shift the superconformal vector ${\tau}^{\textup{ch}}$ of \eqref{eq:standard sconf of chFF} in order to vary the conformal weights of the generators ${\phi}_{\alpha}$ and ${\phi}^{\bar{\alpha}}$ of ${F}^{\textup{ch}}_{N=1}({\mathfrak{g}}, f)$. Note that a 1-parameter deformation of the superconformal vector of the $bc$–${\beta}{\gamma}$ system was introduced in \cite{FGLS92}. The following proposition is an $N = 1$ superfield description of the superconformal vector of the $bc$–${\beta}{\gamma}$ system in \cite{FGLS92}.

\begin{proposition}
Let ${F}^{\textup{ch}}_{N=1}({\mathfrak{g}}, f)$ be a SUSY charged free fermion vertex algebra, and let ${\tau}^{\textup{ch}}$ be the superconformal vector defined in \eqref{eq:standard sconf of chFF}.
Then the following vector of the SUSY charged free fermion vertex algebra
\begin{align}\label{eq:shifted sconf of chFF}
{\tau}^{\textup{ch}}_{\textup{\textbf{m}}} = {\tau}^{\textup{ch}} + {\partial}\Big( \sum\limits_{{\alpha} \in {I}^{+}}{m}_{\alpha}{\phi}_{\alpha}{\phi}^{\bar{\alpha}} \Big),
\end{align}
where ${\textup{\textbf{m}}} = ({m}_{\alpha})_{{\alpha} \in {I}^{+}}$ and ${m}_{\alpha} \in \mathbb{C}$, is also a superconformal vector of ${F}^{\textup{ch}}_{N=1}({\mathfrak{g}}, f)$ with central charge $\sum_{{\alpha} \in {I}^{+}}(6{m}_{\alpha} + 3)$.
\end{proposition}

\begin{proof}
For example, for ${\alpha} \in {I}^{+}_{\bar{0}}$, we have
\begin{align*}
[{\phi}_{\alpha} \ _\Lambda \ ({\partial}{\phi}_{\alpha}){\phi}^{\bar{\alpha}}] 
= [{\phi}_{\alpha} \ _\Lambda \ {\partial}{\phi}_{\alpha}]{\phi}^{\bar{\alpha}} + {\partial}{\phi}_{\alpha}[{\phi}_{\alpha} \ _\Lambda \ {\phi}^{\bar{\alpha}}] + \int_{0}^{\Lambda} [[{\phi}_{\alpha} \ _\Lambda \ {\partial}{\phi}_{\alpha}] \ _{\Gamma} \ {\phi}^{\bar{\alpha}}] \ d\Gamma
= {\partial}{\phi}_{\alpha},
\end{align*}
since, from the sesquilinearity, we know that:
\begin{align*}
[{\phi}_{\alpha} \ _\Lambda \ {\partial}{\phi}_{\alpha}] = ({\partial} + {\lambda})[{\phi}_{\alpha} \ _\Lambda \ {\phi}_{\alpha}] = 0.
\end{align*}
Hence, by the skew-symmetry, we obtain the following relation:
\begin{align*}
[({\partial}{\phi}_{\alpha}){\phi}^{\bar{\alpha}} \ _\Lambda \ {\phi}_{\alpha}] 
= [{\phi}_{\alpha} \ _{- {\nabla} - {\Lambda}} \ ({\partial}{\phi}_{\alpha}){\phi}^{\bar{\alpha}}] 
= {\partial}{\phi}_{\alpha}.
\end{align*}
Similarly, the following $\Lambda$-bracket relations:
\begin{align*}
[({\partial}{\phi}_{\alpha}){\phi}^{\bar{\alpha}} \ _\Lambda \ {\phi}^{\bar{\alpha}}] 
= ({\partial} + {\lambda}){\phi}^{\bar{\alpha}},
{\quad}
[{D}{\phi}_{\alpha}{D}{\phi}^{\bar{\alpha}} \ _\Lambda \ {\phi}_{\alpha}] 
= ({\partial} + {\chi}{D}){\phi}_{\alpha},
{\quad}
[{D}{\phi}_{\alpha}{D}{\phi}^{\bar{\alpha}} \ _\Lambda \ {\phi}^{\bar{\alpha}}] 
= ({\partial} + {\chi}{D}){\phi}^{\bar{\alpha}},
\end{align*}
can be derived using the properties of $N = 1$ SUSY vertex algebras. Therefore, from the non-commutative Wick formula, we have:
\begin{align*}
[({\partial}{\phi}_{\alpha}){\phi}^{\bar{\alpha}} \ _\Lambda \ {\phi}_{\alpha}{\phi}^{\bar{\alpha}}] 
&= [({\partial}{\phi}_{\alpha}){\phi}^{\bar{\alpha}} \ _\Lambda \ {\phi}_{\alpha}]{\phi}^{\bar{\alpha}} + {\phi}_{\alpha}[({\partial}{\phi}_{\alpha}){\phi}^{\bar{\alpha}} \ _\Lambda \ {\phi}^{\bar{\alpha}}] + \int_{0}^{\Lambda}[[({\partial}{\phi}_{\alpha}){\phi}^{\bar{\alpha}} \ _\Lambda \ {\phi}_{\alpha}] \ _{\Gamma} \ {\phi}^{\bar{\alpha}}] \ d\Gamma
\\
&= ({\partial}{\phi}_{\alpha}){\phi}^{\bar{\alpha}} + {\phi}_{\alpha}({\partial} + {\lambda}){\phi}^{\bar{\alpha}} + \int_{0}^{\Lambda}[{\partial}{\phi}_{\alpha} \ _{\Gamma} \ {\phi}^{\bar{\alpha}}] \ d\Gamma 
= ({\partial} + {\lambda})({\phi}_{\alpha}{\phi}^{\bar{\alpha}}),
\\
[{D}{\phi}_{\alpha}{D}{\phi}^{\bar{\alpha}} \ _\Lambda \ {\phi}_{\alpha}{\phi}^{\bar{\alpha}}] 
&= [{D}{\phi}_{\alpha}{D}{\phi}^{\bar{\alpha}} \ _\Lambda \ {\phi}_{\alpha}]{\phi}^{\bar{\alpha}} + {\phi}_{\alpha}[{D}{\phi}_{\alpha}{D}{\phi}^{\bar{\alpha}} \ _\Lambda \ {\phi}^{\bar{\alpha}}] + \int_{0}^{\Lambda}[[{D}{\phi}_{\alpha}{D}{\phi}^{\bar{\alpha}} \ _\Lambda \ {\phi}_{\alpha}] \ _{\Gamma} \ {\phi}^{\bar{\alpha}}] \ d\Gamma
\\
&= \big( ({\partial} + {\chi}{D}){\phi}_{\alpha} \big){\phi}^{\bar{\alpha}} + {\phi}_{\alpha}({\partial} + {\chi}{D}){\phi}^{\bar{\alpha}} + \int_{0}^{\Lambda}[({\partial} + {\chi}{D}){\phi}_{\alpha} \ _{\Gamma} \ {\phi}^{\bar{\alpha}}] \ d\Gamma
\\
&= ({\partial} + {\chi}{D})({\phi}_{\alpha}{\phi}^{\bar{\alpha}}) + {\lambda}{\chi},
\end{align*}
so that we can check the vectors ${\tau}^{\textup{ch}}$ and $(\sum\limits_{{\alpha} \in {I}^{+}}{m}_{\alpha}{\phi}_{\alpha}{\phi}^{\bar{\alpha}})$ satisfy \eqref{eq:SUSY ver shift of sconf} with ${c}_{1} = \sum_{{\alpha} \in {I}^{+}}{m}_{\alpha}$ and ${c}_{2} = 0$.
\end{proof}

\vskip 1mm

Note that there is a unique element $h \in \mathfrak{g}$ in the Cartan subalgebra which induces the Dynkin grading through the eigenvalues for the adjoint action of $\frac{h}{2}$.
Consider the vectors ${\tau}^{\mathfrak{g}}_{h}$ in \eqref{eq:shifted KT} and ${\tau}^{\textup{ch}}_{\textup{\textbf{m}}}$ in \eqref{eq:shifted sconf of chFF}, and define the vector
\begin{align}\label{eq:sconf of C(g,f,k)}
{\tau}^{\mathcal{C}}_{h, {\textup{\textbf{m}}}} := {\tau}^{\mathfrak{g}}_{h} + {\tau}^{\textup{ch}}_{\textup{\textbf{m}}} \in {\mathcal{C}}({\bar{\mathfrak{g}}}, f, k),
\end{align}
where ${\textup{\textbf{m}}} = ({m}_{\alpha})_{{\alpha} \in {I}^{+}}$ for ${m}_{\alpha} \in \mathbb{C}$. Then we have shown the following result:

\begin{theorem}\label{theorem:sconf of tensor of aff and fermion}
The vector ${\tau}^{\mathcal{C}}_{h, {\textup{\textbf{m}}}}$ in \eqref{eq:sconf of C(g,f,k)} is a superconformal vector of the $N = 1$ SUSY vertex algebra ${\mathcal{C}}({\bar{\mathfrak{g}}}, f, k)$ where the central charge is given by
\begin{align*}
{c}^{\mathcal{C}}_{h, {\textup{\textbf{m}}}} = \frac{(3k - 2{h}^{\vee}){\textup{sdim}}{\mathfrak{g}}}{2k} - 3k(h|h) + \sum_{{\alpha} \in {I}^{+}}(6{m}_{\alpha} + 3).
\end{align*}
\end{theorem}

\vskip 1mm

Note that, with respect to the superconformal vector ${\tau}^{\mathcal{C}}_{h, {\textup{\textbf{m}}}}$ in Theorem \ref{theorem:sconf of tensor of aff and fermion}, the vectors ${\phi}_{\alpha}$ and ${\phi}^{\bar{\alpha}}$ for ${\alpha} \in {I}^{+}_{\bar{0}}$ (resp. ${\alpha} \in {I}^{+}_{\bar{1}}$) are primary of conformal weights $-{\frac{1}{2}}{m}_{\alpha}$ and ${\frac{1}{2}}{m}_{\alpha} + {\frac{1}{2}}$ (resp. ${\frac{1}{2}}{m}_{\alpha} + {\frac{1}{2}}$ and $-{\frac{1}{2}}{m}_{\alpha}$). For $a \in {\mathfrak{g}}(i)$, the vector $\bar{a}$ is primary of conformal weight ${\frac{1}{2}} - i$, while the vector $a$ has conformal weight $1 - i$, and $a$ is primary unless $(h|a) \neq 0$, that is:
\begin{align*}
[{\omega}^{\mathcal{C}}_{h, {\textup{\textbf{m}}}} \ _{\lambda} \ {a}] = ({\partial} + (1 - i){\lambda}){a} - {\frac{k}{2}}(h|a){\lambda}^{2},
\end{align*}
where ${\omega}^{\mathcal{C}}_{h, {\textup{\textbf{m}}}}$ is the corresponding conformal vector of ${\tau}^{\mathcal{C}}_{h, {\textup{\textbf{m}}}}$.

\begin{corollary}
The $\Lambda$-brackets between the generators of the $N = 1$ SUSY vertex algebra ${\mathcal{C}}({\bar{\mathfrak{g}}}, {f}, k)$ and the superconformal vector ${\tau}^{\mathcal{C}}_{h, {\textup{\textbf{m}}}}$ are given by the following relations{\textup{:}}
\begin{gather*}
[{\tau}^{\mathcal{C}}_{h, {\textup{\textbf{m}}}} \ _{\Lambda} \ {\bar{a}}] = (2{\partial} + (1-2i){\lambda} + {\chi}{D}){\bar{a}} - k(h|a){\lambda}{\chi},
\\
[{\tau}^{\mathcal{C}}_{h, {\textup{\textbf{m}}}} \ _{\Lambda} \ {\phi}_{\alpha}] = \big( 2{\partial} + {\frac{1}{2}}\big((-1)^{p({\alpha})+1}(2{m}_{\alpha} + 1) + 1\big){\lambda} + {\chi}{D} \big){\phi}_{\alpha},
\\
[{\tau}^{\mathcal{C}}_{h, {\textup{\textbf{m}}}} \ _{\Lambda} \ {\phi}^{\bar{\alpha}}] = \big( 2{\partial} + {\frac{1}{2}}\big((-1)^{p({\alpha})}(2{m}_{\alpha} + 1) + 1\big){\lambda} + {\chi}{D} \big){\phi}^{\bar{\alpha}},
\end{gather*}
for $a \in {\mathfrak{g}}(i)$, ${\alpha} \in {I}^{+}$.
\end{corollary}

\begin{proof}
It follows from direct computations.
\end{proof}

\vskip 6mm

\section{Supersymmetric extension theorem}\label{section5}
\setcounter{equation}{0}

Recall that an $N = 1$ (resp. $N = 2$) SUSY vertex algebra $V$ is a vertex algebra with the $m$-th product for $m\in \mathbb{Z}_+$ defined by
$a_{(m)}b:= a_{(m|1)}b \quad (\, \text{resp.}\ a_{(m)}b:= a_{(m|11)}b \,)$. 
Conversely, in this section, we investigate conditions to embed a given vertex algebra into an $N = 1$ SUSY vertex algebra. 
In Theorem \ref{theorem:N=1 extension}, we also present the construction of an $N = 1$ SUSY extension of $V(R)$ for any Lie conformal algebra $R$, by considering some proper ghost vectors which consist of the SUSY counterpart of vectors in the given Lie conformal algebra. 

\vskip 1mm

In order to impose an $N = 1$ SUSY structure on a vertex algebra $V$, we need an odd derivation $D$ on $V$ which satisfies $[D, D] = 2{\partial}$ and
\begin{align}\label{eq5.1}
D[a \ _{\lambda} \ b] = [Da \ _{\lambda} \ b] + (-1)^{p(a)}[a \ _{\lambda} \ Db],
\end{align}
where $[D,\lambda]=0$ (see Section 5.9 of \cite{Kac98}). In other words, the operator $D$ should also be an odd derivation with respect to the $\lambda$-bracket.
Indeed, by assuming the above properties, the ${\Lambda}$-bracket 
\begin{align}\label{eq5.2}
[a \ _{\Lambda} \ b] := [Da \ _{\lambda} \ b] + {\chi}[a \ _{\lambda} \ b],
\end{align}
defined on $V$ satisfies all the axioms for $N = 1$ SUSY vertex algebras. The sesquilinearity, skew-symmetry and Jacobi identity of $N = 1$ SUSY Lie conformal algebras follow from those properties of non-SUSY Lie conformal algebras, \eqref{eq5.1} and the property $[D, D] = 2{\partial}$. The quasi-commutativity, quasi-associativity and non-commutative Wick formula of $N = 1$ SUSY vertex algebras also follow from those properties of vertex algebras. For example, the non-commutative Wick formula of $N = 1$ SUSY vertex algebras holds since we can check the following computations:
\begin{align*}
[a \ _{\Lambda} \ bc] &= [Da \ _{\lambda} \ bc] + {\chi}[a \ _{\lambda} \ bc]
\\ 
&  = \big( \, [Da \ _{\lambda} \ b]c + (-1)^{(p(a)+1)p(b)}b[Da \ _{\lambda} \ c] + \int_{0}^{\lambda} [[Da \ _{\lambda} \ b] \ _{\gamma} \ c] \ d\gamma\, \big)
\\
& \quad + {\chi} \,  \big( \,[a \ _{\lambda} \ b]c + (-1)^{p(a)p(b)}b[a \ _{\lambda} \ c] + \int_{0}^{\lambda} [[a \ _{\lambda} \ b] \ _{\gamma} \ c] \ d\gamma\, \big)
\\
& =  \big( \, [Da \ _{\lambda} \ b]c + {\chi}[a \ _{\lambda} \ b]c \, \big)
+ (-1)^{(p(a)+1)p(b)} \big( \, b[Da \ _{\lambda} \ c] + (-1)^{p(b)}{\chi}b[a \ _{\lambda} \ c] \, \big)
\\
& \quad + {\partial}_{\eta}\int_{0}^{\lambda} [D[Da \ _{\lambda} \ b] \ _{\gamma} \ c] + {\eta}[[Da \ _{\lambda} \ b] \ _{\gamma} \ c] - {\chi}[D[a \ _{\lambda} \ b] \ _{\gamma} \ c] - {\chi}{\eta}[[a \ _{\lambda} \ b] \ _{\gamma} \ c] \ d\gamma
\\
& = [a \ _{\Lambda} \ b]c + (-1)^{(p(a)+1)p(b)}b[a \ _{\Lambda} \ c]
+ \int_{0}^{\Lambda} [[a \ _{\Lambda} \ b] \ _{\Gamma} \ c] \ d\Gamma,
\end{align*}
for the $\Lambda$-bracket in \eqref{eq5.2} and  $a,b,c \in V$. Hence if there is an odd derivation $D$ of the normally ordered product satisfying $[D, D] = 2{\partial}$ and \eqref{eq5.1}, then $V$ itself can be considered as an $N=1$ SUSY vertex algebra.

\vskip 1mm
 
Now, let us discuss the case when we cannot find such an odd derivation $D$. In this case, we aim to see if $V$ can be embedded into an $N = 1$ SUSY vertex algebra. More concretely, the following question is the main topic we will deal with in the rest of Section \ref{section5}.

\begin{question} \label{question:N=0 to N=1}
\rm
Consider a Lie conformal algebra
\begin{equation} \label{eq:sec 5,LCA}
{R} = {\mathbb{C}}[{\partial}] \otimes {\textup{span}}_{\mathbb{C}}\{ {u}_{i} \}_{ i {\in} I } \oplus E_{C},
\end{equation}
where the set $\{ {u}_{i} \}_{ i {\in} I }$ is linearly independent and $E_C$ is the central extension part of $R$ with $\partial E_C=0$. Can we find a Lie conformal algebra $\widetilde{R}$ containing $R$ such that the corresponding universal enveloping vertex algebra $V(\widetilde{R})$ is an $N = 1$ SUSY vertex algebra?
\end{question}

\vskip 1mm

If there is an $N = 1$ SUSY vertex algebra $V(\widetilde{R})$ containing the original vertex algebra $V(R)$ as a vetex subalgebra, then we say $V(\widetilde{R})$ is an {\it $N=1$ SUSY extension of $V(R)$}.
Suppose $R$ is the  Lie conformal algebra and $I$ is the index set in \eqref{eq:sec 5,LCA}. For subindex sets $J, \bar{J} \subset I$ such that  $I=J \sqcup \bar{J}$, we consider new data ${\bar{R}}$, $D$ and $V({\widetilde{R}})$ for $\widetilde{R}= R \oplus \bar{R}$ which are given as follows:
\begin{enumerate}[]
\item (D-1) \ ${\bar{R}} = {\mathbb{C}}[{\partial}] \otimes {\textup{span}}_{\mathbb{C}}\{ {v}_{i} \}_{ i {\in} I }$ is a $\mathbb{C}[\partial]$-module and $p({v}_{i}) = 1 - p({u}_{i})$ for any $i \ {\in} \ I$.
\item (D-2) \ $D$ is an odd endomorphism on $\widetilde{R}={R} \oplus {\bar{R}}$, such that $[D, D] = 2{\partial}$ and 
$D({u}_{j}) = {v}_{j}$, $D({v}_{\bar{\jmath}})= {u}_{\bar{\jmath}}$ \text{ for } $j \in J$, $\bar{\jmath} \in \bar{J}$, and $D(E_C)=0$. For the simplicity of notations, let us denote  ${\bar{v}}_{j} := {u}_{j}$ and $ {\bar{u}}_{\bar{\jmath}} := {v}_{\bar{\jmath}}$. Then we have \[D(\bar{u})=u \quad \text{ for } \quad  u \in \{ v_j \, | \, j \in J \} \cup \{ u_{\bar{\jmath}} \, | \, \bar{\jmath} \in \bar{J} \}.\]
\item (D-3) \ Take an ordered ${\mathbb{C}}$-basis of $R$ and extend it to an ordered basis $\{{r}_{i}\}_{i {\in} {\mathcal{I}}}$ of $\widetilde{R}$. Let $V(\widetilde{R})$ be a $\mathbb{C}$-algebra endowed with a basis 
\[ \{ {r}_{i_1}{r}_{i_2} \cdots {r}_{i_l} \, | \, l \ {\in} \ {\mathbb{Z}_+}, \ {i}_{1} \leq {i}_{2} \leq \cdots \leq {i}_{l}, \ {i}_{t} < {i}_{t+1} \text{ if } p({r}_{i_t}) = 1 \}, \]
where if $l = 0$ then ${r}_{i_1}{r}_{i_2} \cdots {r}_{i_l} := 1$. Moreover, we assume $V(\widetilde{R})$ is the differential algebra with the odd derivation $D$ induced from $D$ in (D-2).
\end{enumerate}

\vskip 1mm

In the rest of this section, we will look for assumptions to give an $N = 1$ SUSY vertex algebra structure on  $V(\widetilde{R})$. 
The following example shows the $bc$–$\beta\gamma$ system is an $N=1$ SUSY extension of the ${\beta}{\gamma}$ system (see Example 5.12 of \cite{HK07}).

\begin{example}[${\beta}{\gamma}$ system vs $bc$–${\beta}{\gamma}$ system] \label{ex:sec5, bcbetagamma}
\rm
Let ${R} = {\mathbb{C}}[{\partial}] \otimes {\textup{span}}_{\mathbb{C}}\{ \beta, \gamma \} \oplus \mathbb{C}C$ be the Lie conformal algebra for the ${\beta}{\gamma}$ system and let ${\widetilde{R}} = R \oplus \bar{R} = {\mathbb{C}}[{\partial}] \otimes {\textup{span}}_{\mathbb{C}}\{ \beta, \gamma, b, c \} \oplus \mathbb{C}C$ be the Lie conformal algebra for the $bc$–$\beta\gamma$ system in Example \ref{ex:def of bcbetagamma}, where ${\bar{R}} = {\mathbb{C}}[{\partial}] \otimes {\textup{span}}_{\mathbb{C}}\{ b, c \}$.
Then the $bc$–$\beta\gamma$ system is an $N = 1$ SUSY extension of the $\beta\gamma$ system (see Example \ref{example: SUSY bcbetagamma}), where the odd derivation $D$ is defined by
\begin{align*}
D(b) = {\beta},
{\quad}
D(c) = {\partial}{\gamma},
{\quad}
D({\beta}) = {\partial}{b},
{\quad}
D({\gamma}) = c,
\end{align*}
i.e.,  $b= \bar{\beta}$ and $\gamma= \bar{c}$. The $\lambda$-bracket on $V(\widetilde{R})$ can be rewritten via a $\Lambda$-bracket as follows:
\[ [b \ _{\Lambda} \ {\gamma}] = [{\gamma} \ _{\Lambda} \ b] = C. \]
\end{example}

\vskip 1mm

To deal with general cases, observe the sesquilinearity of an $N = 1$ SUSY vertex algebra $V$ in terms of the $\lambda$-bracket 
\[ [a \ _\lambda \ b ] = \sum_{m \in \mathbb{Z}_+} \frac{{\lambda}^{m}}{m!} a_{(m|1)}b\quad \text{ for } \quad a,b \in V.\]
Note that $V$ is a vertex algebra with respect to the above $\lambda$-bracket. Consider ${\bar{a}}$ and ${\bar{b}}$ in $V$ such that $D{\bar{a}} = a$ and $D{\bar{b}} = b$. Then the sesquilinearity of $N = 1$ SUSY Lie conformal algebras for the pair (${\bar{a}}$, ${\bar{b}}$) tells that
$[a \ _{\lambda} \ b] = (-1)^{p(a)}(D[a \ _{\lambda} \ \bar{b}] + {\lambda}[\bar{a} \ _{\lambda} \ \bar{b}])$, and  $[\bar{a} \ _{\lambda} \ b] = (-1)^{p(a)}([a \ _{\lambda} \ \bar{b}] - D[\bar{a} \ _{\lambda} \ \bar{b}])$. Equivalently, we have 
\begin{equation*}
\begin{aligned}
& \text{ (SEF1) \quad $[D{\bar{a}} \ _{\lambda} \ D{\bar{b}}] = (-1)^{p({\bar{a}})+1}(D[D{\bar{a}} \ _{\lambda} \ \bar{b}] + {\lambda}[\bar{a} \ _{\lambda} \ \bar{b}]),$}
\\
& \text{ (SEF2) \quad $[\bar{a} \ _{\lambda} \ D{\bar{b}}] = (-1)^{p({\bar{a}})+1}([D{\bar{a}} \ _{\lambda} \ \bar{b}] - D[\bar{a} \ _{\lambda} \ \bar{b}]),$ }
\end{aligned}
\end{equation*}
which are called {\it supersymmetric extension formulas}. These formulas are crucially needed to see if the differential algebra $V(\widetilde{R})$ in (D-3) is indeed an $N = 1$ SUSY vertex algebra.

\vskip 1mm

\begin{definition}\label{definition:SUSY generator}
\rm
Let $R$ be a ${\mathbb{C}}[{\partial}]$-module with a ${\lambda}$-bracket. For a subset $S \subseteq R$ and an odd endomorphism $D \ {\in} \ \textup{End}(R)$ such that $[D, D] = 2{\partial}$, we say the pair ($S$, $D$) \textit{satisfies the supersymmetric extension formulas}, if (SEF1) and (SEF2) hold for all ${\bar{a}}, {\bar{b}} \ {\in} \ S$. In this case, we also call the subset $S$ a \textit{supersymmetric generator in the $D$–direction}.
\end{definition}

\vskip 1mm

Under the assumptions in Definition \ref{definition:SUSY generator}, note that if an ordered pair $({\bar{a}}, {\bar{b}})$ satisfies (SEF1) and (SEF2), then the odd endomorphism $D$ acts as an odd derivation on the following ${\lambda}$-brackets:
\begin{align*}
[{\bar{a}} \ _{\lambda} \ {\bar{b}}],
{\quad}
[D{\bar{a}} \ _{\lambda} \ {\bar{b}}],
{\quad}
[{\bar{a}} \ _{\lambda} \ D{\bar{b}}],
{\quad}
[D{\bar{a}} \ _{\lambda} \ D{\bar{b}}],
\end{align*}
since we know that the following equalities hold:
\begin{equation}\label{eq:D-derivation wrt lambda}
\begin{aligned}
& D[{\bar{a}} \ _{\lambda} \ D{\bar{b}}] 
= (-1)^{p(a)}(D[D{\bar{a}} \ _{\lambda} \ {\bar{b}}] - {\partial}[{\bar{a}} \ _{\lambda} \ {\bar{b}}]) 
\\
& \hskip 5mm= [D{\bar{a}} \ _{\lambda} \ D{\bar{b}}] + (-1)^{p(a)+1}({\partial} + {\lambda})[{\bar{a}} \ _{\lambda} \ {\bar{b}}] 
= [D{\bar{a}} \ _{\lambda} \ D{\bar{b}}] + (-1)^{p({\bar{a}})}[{\bar{a}} \ _{\lambda} \ {\partial}{\bar{b}}],
\\
& D[D{\bar{a}} \ _{\lambda} \ D{\bar{b}}] 
= (-1)^{p(a)}({\partial}[D{\bar{a}} \ _{\lambda} \ {\bar{b}}] + {\lambda}{D}[{\bar{a}} \ _{\lambda} \ {\bar{b}}]) 
\\
&\hskip 5mm = - {\lambda}[{\bar{a}} \ _{\lambda} \ D{\bar{b}}] + (-1)^{p(a)}({\partial} + {\lambda})[D{\bar{a}} \ _{\lambda} \ {\bar{b}}] 
= [{\partial}{\bar{a}} \ _{\lambda} \ D{\bar{b}}] + (-1)^{p(a)}[D{\bar{a}} \ _{\lambda} \ {\partial}{\bar{b}}].
\end{aligned}
\end{equation}
Furthermore,  Proposition \ref{proposition:N=1 extension strongly generated} shows the two formulas (SEF1) and (SEF2) are sufficient to give solutions to Question \ref{question:N=0 to N=1}.

\begin{proposition}\label{proposition:N=1 extension strongly generated}
Let ${R} = {\mathbb{C}}[{\partial}] \otimes {\textup{span}}_{\mathbb{C}}\{ {u}_{i} \}_{ i {\in} I }\oplus E_C$ be a Lie conformal algebra in \eqref{eq:sec 5,LCA}, and let $J$ and $\bar{J}$ be subsets of $I$ such that $I=J \sqcup \bar{J}$. Suppose ${\widetilde{R}}$ and $D$ are given by the data in ${\textup{(D-1)}}$ and ${\textup{(D-2)}}$, and the ${\lambda}$-brackets between any two elements in $\{{u}_{i}, {v}_{i}\}_{i {\in} I}$ are given so that ${\widetilde{R}}$ is a Lie conformal algebra. If $S=\{u_i= {\bar{v}}_{i} \, | \, i\in J\}\cup\{v_i={\bar{u}}_{i} \, | \, i\in \bar{J}\} $ is a supersymmetic generator with the $D$–direction,  then $V(\widetilde{R})$ is an $N = 1$ SUSY vertex algebra with the odd derivation $D$.
\end{proposition}

\begin{proof}
From \eqref{eq:D-derivation wrt lambda}, the odd endomorphism $D$ is indeed an odd derivation of the ${\lambda}$-brackets of ${\widetilde{R}}$.
Moreover, from the non-commutative Wick formula, the following equality holds:
\begin{align*}
D[a \ _{\lambda} \ bc]
&= D\big([a \ _{\lambda} \ b]c + (-1)^{p(a)p(b)}b[a \ _{\lambda} \ c] + \int_{0}^{\lambda} [[a \ _{\lambda} \ b] \ _{\gamma} \ c] \ d\gamma \big)
\\
&= (D[a \ _{\lambda} \ b])c + (-1)^{p(a)+p(b)}[a \ _{\lambda} \ b]Dc
\\
&{\quad} + (-1)^{p(a)p(b)}(Db)[a \ _{\lambda} \ c] + (-1)^{(p(a)+1)p(b)}bD[a \ _{\lambda} \ c] + \int_{0}^{\lambda} D[[a \ _{\lambda} \ b] \ _{\gamma} \ c] \ d\gamma
\\
&= [Da \ _{\lambda} \ b]c + (-1)^{p(a)}[a \ _{\lambda} \ Db]c + (-1)^{p(a)+p(b)}[a \ _{\lambda} \ b]Dc
\\
&{\quad} + (-1)^{p(a)p(b)}(Db)[a \ _{\lambda} \ c] + (-1)^{(p(a)+1)p(b)}b[Da \ _{\lambda} \ c] + (-1)^{p(a)p(b)+p(a)+p(b)}b[a \ _{\lambda} \ Dc]
\\
&{\quad} + \int_{0}^{\lambda} [[Da \ _{\lambda} \ b] \ _{\gamma} \ c] \ d\gamma + (-1)^{p(a)}\int_{0}^{\lambda} [[a \ _{\lambda} \ Db] \ _{\gamma} \ c] \ d\gamma + (-1)^{p(a)+p(b)}\int_{0}^{\lambda} [[a \ _{\lambda} \ b] \ _{\gamma} \ Dc] \ d\gamma
\\
&= [Da \ _{\lambda} \ bc] + (-1)^{p(a)}[a \ _{\lambda} \ (Db)c] + (-1)^{p(a)+p(b)}[a \ _{\lambda} \ bDc]
\\
&= [Da \ _{\lambda} \ bc] + (-1)^{p(a)}[a \ _{\lambda} \ D(bc)],
\end{align*}
for $a, b, c \in \widetilde{R}$.
Hence, $D$ is an odd derivation of the ${\lambda}$-brackets on $V({\tilde{R}})$ so that $V({\tilde{R}})$ is an $N = 1$ SUSY vertex algebra with the ${\Lambda}$-brackets in \eqref{eq5.2}.
\end{proof}

\vskip 1mm

In Proposition \ref{proposition:N=1 extension strongly generated}, we find an $N=1$ SUSY Lie conformal algebra structure on $\widetilde{R}$ properly concatenating the $\mathbb{C}[\partial]$-module $\bar{R}$ to the Lie conformal algebra $R$ which yields the $N = 1$ SUSY extension $V(\widetilde{R})$ of $V(R)$. In this process, (SEF1) and (SEF2) are essential in the sense that they guarantee the $N=1$ supersymmetry of $V(\widetilde{R})$. In the following example, we present another fundamental example for Question \ref{question:N=0 to N=1} when $V(R)$ is the affine vertex algebra (see also Example 5.9 of \cite{HK07}).

\begin{example} [Affine vertex algebra vs $N = 1$ SUSY affine vertex algebra] \label{ex:N=1 affine}
\rm
Let $\mathfrak{g}$ be a Lie superalgebra with an even invariant supersymmetric bilinear form $(\, |\, )$. Consider the current Lie conformal algebra ${\textup{Cur}\mathfrak{g}} = \mathbb{C}[\partial] \otimes \mathfrak{g}\oplus \mathbb{C}K$ introduced in Example \ref{example: affine VA} and let $\widetilde{{\textup{Cur}\mathfrak{g}}} = \mathbb{C}[\partial] \otimes (\mathfrak{g} \oplus \bar{\mathfrak{g}}) \oplus \mathbb{C}K$ be a $\mathbb{C}[\partial]$-module, where $\bar{\mathfrak{g}}$ is the parity reversed vector superspace of $\mathfrak{g}$, i.e. $\bar{\mathfrak{g}}=\{ \bar{a} \, | \, a\in \mathfrak{g}, \,  p(\bar{a})= 1- p(a) \}$. Define the $\lambda$-bracket on $\widetilde{{\textup{Cur}\mathfrak{g}}}$ by
\begin{align}\label{eq:SUSY affine lambda bracket}
[a \ _{\lambda} \ b] = [a, b] + K{\lambda}(a|b),
{\quad}
[a \ _{\lambda} \ {\bar{b}}] = (-1)^{p(a)}{\overline{[a,b]}},
{\quad}
[{\bar{a}} \ _{\lambda} \ {\bar{b}}] = (-1)^{p(a)}K(a|b),
\end{align}
and the odd derivation $D$ on $\widetilde{{\textup{Cur}\mathfrak{g}}}$ by $D{\bar{a}} = a$ and $Da = {\partial}{\bar{a}}$ for $a, b \in \mathfrak{g}$ (see Example 2.6 of \cite{HZ10}). Then the ${\lambda}$-brackets in \eqref{eq:SUSY affine lambda bracket} satisfy (SEF1) and (SEF2) and thus $\bar{\mathfrak{g}}$ is a supersymmetric generator in the $D$–direction. In terms of $\Lambda$-bracket, the ${\lambda}$-bracket \eqref{eq:SUSY affine lambda bracket} on $V(\widetilde{{\textup{Cur}\mathfrak{g}}})$ can be rewritten as
\begin{align*}
[\bar{a} \ _\Lambda \ \bar{b}] = (-1)^{p(a)}\big(\overline{[a, b]} + K\chi (a|b) \big),
\end{align*}
and the quotient vertex algebra $V(\widetilde{{\textup{Cur}\mathfrak{g}}})/\left< K - k \right>$ is the $N=1$ SUSY affine vertex algebra $V^k_{N=1}(\mathfrak{g})$ of level $k$ introduced in Example \ref{example:SUSY affine VA}. 
\end{example}

\vskip 1mm

Recall, in Example \ref{ex:sec5, bcbetagamma}, that $R$ is the Lie conformal algebra generated by $\beta$, $\gamma$ and the central element $C$ whose universal enveloping vertex algebra $V(R)$ is the ${\beta}{\gamma}$ system  and 
the $bc$–$\beta\gamma$ system $V(\widetilde{R})$ is an $N=1$ SUSY extension of $V(R)$.
In this case, the choice of $\mathbb{C}[\partial]$-module $\bar{R}$  and derivation $D$ is depicted by the diagram (Case 1 : $bc$–$\beta\gamma$ system).
On the other hand, suppose we consider another $\mathbb{C}[\partial]$-module $\bar{R}$ generated by $\bar{\beta}$ and $\bar{\gamma}$, as in the diagram (Case 2 : No supersymmetry). Then there is no supersymmetry in $V(\widetilde{R})$.
\vskip 2mm
\[
\begin{tikzcd}[row sep=scriptsize, column sep=scriptsize, style={font=\small}]
{{\bar{u}_{1}} = b} \arrow{rrr}{D} &&& {{\beta} = {u}_{1}} &&& {} &&&& {{\bar{u}_{1}} = {\bar{\beta}}} \arrow{rrr}{D} &&& {{\beta} = {u}_{1}}
\\
{} &&& {{u}_{2} = {\gamma}} \arrow{rrr}{D} &&& {c = D({u}_{2})} &&&& {{\bar{u}}_{2} = {\bar{\gamma}}} \arrow{rrr}{D} &&& {{\gamma} = {u}_{2}}
\\
\\
{} \arrow[phantom, uu, ""{name=pht1}] &&& {} &&& {} \arrow[phantom, uu, ""{name=pht2}] &&&& {} \arrow[phantom, uu, ""{name=pht3}] &&& {} \arrow[phantom, uu, ""{name=pht4}]
\arrow[phantom, from=pht1 ,to=pht2,"{(\text{Case 1 : $bc$–$\beta\gamma$ system})}"]
\arrow[phantom, from=pht3 ,to=pht4,"{(\text{Case 2 : No supersymmetry}})"]
\end{tikzcd}
\]
That is because, in this case, (SEF1) does not have a solution. For example, for ${u}_{1} = {\beta}$, ${u}_{2} = {\gamma}$, we cannot find the   ${\lambda}$-brackets $[{u_1} \ _{\lambda} \ {\bar{{u}}_{2}}]$ and $[{\bar{{u}}_{1}} \ _{\lambda} \ {\bar{{u}}_{2}}]$ satisfying the equation below:
\begin{align*}
C = D[{u_1} \ _{\lambda} \ {\bar{{u}}_{2}}] + {\lambda}[{\bar{{u}}_{1}} \ _{\lambda} \ {\bar{{u}}_{2}}].
\end{align*}

\vskip 1mm

From the observation above, for a given Lie conformal algebra $R$ in \eqref{eq:sec 5,LCA}, the choice of $\bar{R}$ in which $\widetilde{R}= R\oplus \bar{R}$  and $D$ is a crucial part to give an answer to Question \ref{question:N=0 to N=1}.
The subsequent theorem (Theorem \ref{theorem:N=1 extension}) shows the universal enveloping vertex algebra $V(R)$ can be extended to an $N = 1$ SUSY vertex algebra using $\bar{R}$ and $D$ given as follows:
\begin{enumerate}[]
\item (D-2$'$) \ $\bar{R}= \mathbb{C}[\partial]\otimes {\textup{span}}_{\mathbb{C}}\{ {\bar{u}}_{i} \}_{ i {\in} I } \oplus \bar{E}_C$ is the parity reversed $\mathbb{C}[\partial]$-module of $R$, where $\bar{E}_C=\{\bar{C} \, | \, C \in E_C\}$. The odd derivation $D$ is defined by $D({\bar{u}}_{i}) = {u}_{i}$ for $i \in I$, and $D(\bar{C})= C$ for $C \in {E}_{C}$. In addition, ${\partial}{\bar{E}_C} = 0$.
\end{enumerate}

\begin{theorem}\label{theorem:N=1 extension}
Let ${R} = {\mathbb{C}}[{\partial}] \otimes {\textup{span}}_{\mathbb{C}}\{ {u}_{i} \}_{ i {\in} I }\oplus E_C$ be a Lie conformal algebra in \eqref{eq:sec 5,LCA}. Then
$\widetilde{R}=R\oplus \bar{R}$ for the $\mathbb{C}[\partial]$-module $\bar{R}$ given by {\textup{(D-2$'$)}} is a Lie conformal algebra endowed with the $\lambda$-bracket defined by
\begin{gather} \label{eq:N=1 extension bracket}
[{\bar{{u}}_{i}} \ _{\lambda} \ {u}_{j}] := \overline{[{u}_{i} \ _{\lambda} \ {u}_{j}]},
\quad
[{u}_{i} \ _{\lambda} \ {\bar{{u}}_{j}}] := (-1)^{p({u}_{i})}\overline{[{u}_{i} \ _{\lambda} \ {u}_{j}]},
\quad
[{\bar{{u}}_{i}} \ _{\lambda} \ {\bar{{u}}_{j}}] := 0,
\end{gather}
and $[{C} \ _{\lambda} \ {\widetilde{R}}] = [{\bar{C}} \ _{\lambda} \ {\widetilde{R}}] = 0$ for $C\in E_C$.
Moreover, $V(\widetilde{R})$ is an $N = 1$ SUSY extension of $V(R)$.
\end{theorem}

\begin{proof}
For simplicity, we use the notation
$\overline{[a \ _{\lambda} \ b]} = \sum_{m \in {\mathbb{Z}}_{+}} {\frac{{\lambda}^{m}}{m!}}\overline{{a}_{(m)}{b}}$ for $a, b \in R$, where $\overline{{\partial}^{s}{u}_{i}} = {\partial}^{s}{\bar{{u}}_{i}}$.
Then the skew-symmetry and the Jacobi identity of the $\lambda$-bracket can be verified by direct computations. For example, the Jacobi identity for the triple  $\{(u_i, u_j,\bar{u}_k) \, | \, i, j, k \in I \} $ can be checked as follows:
\begin{align*}
[{u}_{i} \ _{\lambda} \ [{u}_{j} \ _{\gamma} \ {\bar{u}}_{k}]] &= (-1)^{p({u}_{j})}[{u}_{i} \ _{\lambda} \ {\overline{[{u}_{j} \ _{\gamma} \ {u}_{k}]}}]
\\
&= (-1)^{p({u}_{i})+p({u}_{j})}{\overline{[{u}_{i} \ _{\lambda} \ [{u}_{j} \ _{\gamma} \ {u}_{k}]]}}
\\
&= (-1)^{p({u}_{i})+p({u}_{j})}{\overline{[[{u}_{i} \ _{\lambda} \ {u}_{j}] \ _{{\lambda} + {\gamma}} \ {u}_{k}]}} + (-1)^{p({u}_{i})p({u}_{j}) + p({u}_{i}) + p({u}_{j})}{\overline{[{u}_{j} \ _{\gamma} \ [{u}_{i} \ _{\lambda} \ {u}_{k}]]}}
\\
&= [[{u}_{i} \ _{\lambda} \ {u}_{j}] \ _{{\lambda} + {\gamma}} \ {\bar{u}}_{k}] + (-1)^{p({u}_{i})p({u}_{j})}[{u}_{j} \ _{\gamma} \ [{u}_{i} \ _{\lambda} \ {\bar{u}}_{k}]].
\end{align*}
Moreover, one can check the pair ($\{{\bar{u}}_{i}\}_{i {\in} I}$, $D$) satisfies the SUSY extension formulas. Hence, by  Proposition \ref{proposition:N=1 extension strongly generated}, the vertex algebra $V({\widetilde{R}}) = V({R}) \otimes V({\bar{R}})$ has an $N = 1$ supersymmetry.
\end{proof}

\vskip 1mm

The $\lambda$-bracket \eqref{eq:N=1 extension bracket} is inspired from the $\Lambda$-bracket of the $N = 1$ SUSY affine vertex algebra (see also Section 3.2.1 of \cite{CFRS90}).
In terms of $\Lambda$-bracket, the bracket \eqref{eq:N=1 extension bracket} can be written as follows:
\begin{align*}
[\bar{a} \ _\Lambda \ \bar{b}] = (-1)^{p(a)}\overline{[a \ _\lambda \ b]}
\end{align*}
for $a, b \in R$.
Indeed, if $R$ is a centerless current Lie conformal algebra, the $N = 1$ SUSY extension obtained from Theorem \ref{theorem:N=1 extension} is the $N = 1$ SUSY affine vertex algebra of level $k = 0$ (see Example \ref{example:SUSY affine VA}). In the following examples, we find an $N = 1$ SUSY extension of the $\beta\gamma$ system and the Virasoro vertex algebra using Theorem \ref{theorem:N=1 extension}.

\begin{example}[$N=1$ SUSY extension of $\beta\gamma$ system]\label{ex: betagamma N=1 ext-2}
\rm
Let ${R} = {\mathbb{C}}[{\partial}] \otimes {\textup{span}}_{\mathbb{C}}\{ \beta, \gamma \} \oplus \mathbb{C}C$ be the Lie conformal algebra for the ${\beta}{\gamma}$ system (see Example \ref{ex:def of bcbetagamma}). If we apply Theorem \ref{theorem:N=1 extension} to $R$ by taking ${u}_{1} = {\beta}$ and ${u}_{2} = {\gamma}$, we know that $\widetilde{R} = {\mathbb{C}}[{\partial}] \otimes {\textup{span}}_{\mathbb{C}}\{ {\beta}, {\gamma}, {\bar{\beta}}, {\bar{\gamma}} \} \oplus {\mathbb{C}}C \oplus {\mathbb{C}}{\bar{C}}$ with the remaining non-zero ${\lambda}$-brackets
\begin{align*}
[{\beta} \ _{\lambda} \ {\bar{\gamma}}] = [{\bar{\beta}} \ _{\lambda} \ {\gamma}] = {\bar{C}},
\end{align*}
is an $N = 1$ SUSY extension of the ${\beta}{\gamma}$ system. Note that $V(\widetilde{R})$ is a distinct vertex algebra from the $bc$–$\beta\gamma$ system, which is another SUSY extension introduced in Example \ref{ex:sec5, bcbetagamma}.
\end{example}

\begin{example}[$N=1$ SUSY extension of Virasoro vertex algebra] \label{ex:N=1 ext of vir}
\rm
Consider the {\it centerless Virasoro Lie conformal algebra} $\textup{Vir} = {\mathbb{C}}[\partial] \otimes {\mathbb{C}}L$ with the ${\lambda}$-bracket $[L \ _\lambda \ L] = ({\partial} + 2{\lambda})L$, and its universal enveloping vertex algebra ${\textup{Vir}}^{0}$ (see Example \ref{example: Virasoro LCA}). First of all, for the $N = 1$ super-Virasoro Lie conformal algebra $\textup{SVir}= \mathbb{C}[\partial]\otimes (\mathbb{C}L \oplus \mathbb{C}G)$
in which the ${\lambda}$-brackets are defined as
\begin{align*}
[L \ _\lambda \ G] = ({\partial} + {\frac{3}{2}}{\lambda})G,
{\quad}
[G \ _\lambda \ G] = 2L,
\end{align*}
if we define the odd derivation $D$ by $DG = 2L$ then $(\{ G \},D)$ satisfies the SUSY extension formulas. Hence the $N=1$ super-Virasoro vertex algebra ${\textup{SVir}}^{0}$ is an $N = 1$ SUSY extension of ${\textup{Vir}}^{0}$. On the other hand, by using the Theorem \ref{theorem:N=1 extension}, the universal enveloping vertex algebra $V({\widetilde{\textup{Vir}}})$ for ${\widetilde{\textup{Vir}}}={\mathbb{C}}[{\partial}] \otimes ({\mathbb{C}}L \oplus {\mathbb{C}}{\bar{L}})$ is also an $N = 1$ SUSY extension of ${\textup{Vir}}^{0}$, where the ${\lambda}$-brackets of ${\widetilde{\textup{Vir}}}$ are given by
\begin{align*}
[{L} \ _\lambda \ {\bar{L}}] = ({\partial} + 2{\lambda}){\bar{L}},
{\quad}
[{\bar{L}} \ _\lambda \ {\bar{L}}] = 0.
\end{align*}
Note that, as in Example 2.7 of \cite{HK07}, if we take an another generating set for the $N = 2$ super-Virasoro vertex algebra by defining $L' := L - {\frac{1}{2}}{\partial}J$, then we know that especially the ${\lambda}$-brackets for $L'$ and ${G}^{+}$ are 
\begin{align*}
[{L'} \ _\lambda \ {L'}] = ({\partial} + 2{\lambda}){L'},
{\quad}
[{L'} \ _\lambda \ {G}^{+}] = ({\partial} + 2{\lambda}){G}^{+},
{\quad}
[{G}^{+} \ _\lambda \ {G}^{+}] = 0.
\end{align*}
Hence, the vertex algebra $V({\widetilde{\textup{Vir}}})$ can be viewed as a vertex subalgebra of the $N = 2$ super-Virasoro vertex algebra of any central charge $c \in \mathbb{C}$.
\end{example}

\vskip 1mm

In Example \ref{ex:N=1 ext of vir}, we introduced two distinct $N=1$ SUSY vertex algebras ${\textup{SVir}}^{0}$ and $V(\widetilde{{\textup{Vir}}})$ which are strongly generated by two elements $L$ and $D^{-1}(L)$. 
Furthermore, an $N=1$ SUSY extension of ${\textup{Vir}}^{0}$ strongly generated by $L$ and a primary field $\bar{L}$ such that $D(\bar{L})=L$ is one of the two vertex algebras mentioned above. To check the detail, 
suppose the conformal weight of $\bar{L}$ is $\Delta \in \mathbb{C}$, i.e. $[L \ _\lambda \ {\bar{L}}] = ({\partial} + {\Delta}{\lambda}){\bar{L}}$. Then, from the equation (SEF1), we have
\begin{align*}
[{\bar{L}} \ _\lambda \ {\bar{L}}] = {\lambda}^{-1}([{L} \ _\lambda \ {L}] - D[{L} \ _\lambda \ {\bar{L}}]) = (2-{\Delta})L.
\end{align*}
Moreover, by using the Jacobi identity, the possible values of $\Delta$ is $\frac{3}{2}$ or $2$ since
\begin{align*}
(2 - {\Delta})({\partial} + 2{\lambda})L &= [{L} \ _\lambda \ [{\bar{L}} \ _\gamma \ {\bar{L}}]]
\\
&= [[{L} \ _\lambda \ {\bar{L}}] \ _{{\lambda} + {\gamma}} \ {\bar{L}}] + [{\bar{L}} \ _\gamma \ [{L} \ _\lambda \ {\bar{L}}]]
\\
&= (({\Delta} - 1){\lambda} - {\gamma})(2 - {\Delta})L + ({\partial} + {\Delta}{\lambda} + {\gamma})(2 - {\Delta})L
\\
&= (2 - {\Delta})({\partial} + (2{\Delta}-1){\lambda})L.
\end{align*}

\vskip 1mm

As the Virasoro vertex algebra ${\textup{Vir}}^{c}$ can be extended to the $N = 1$ super-Virasoro vertex algebra ${\textup{SVir}}^{c}$, it is natural to ask if a vertex algebra can be extended to an $N = 1$ superconformal vertex algebra. In the following Remark \ref{remark:superconformal current}, using the SUSY extension formulas, one can reproduce the $N = 1$ superconformal extension of the affine vertex algebra of level $0$, which is called the {\it superconformal current algebra} in \cite{KT85}. A free field realization and a supersymmetric ghost realization of the superconformal current algebra can be found in \cite{IR94}.

\begin{remark}[\cites{KT85}] \label{remark:superconformal current}
\rm
Let $\mathfrak{g}=\bigoplus_{\omega \in \Omega}\mathfrak{g}_{\omega}$ be an $\Omega$-graded Lie superalgebra. Then, the $N = 1$ SUSY affine vertex algebra $V^{0}_{N = 1}(\mathfrak{g})$ can be embedded into an $N = 1$ superconformal vertex algebra with the conformal weight $\Delta_a=1-\omega$, where $a \in \mathfrak{g}_{\omega}$ is a homogeneous element for $\omega \in \Omega$. More precisely, consider the Lie conformal algebra ${\textup{Cur}\mathfrak{g}}_{N=1} = \mathbb{C}[\partial] \otimes ({\mathfrak{g}} \oplus {\bar{\mathfrak{g}}})$ with the ${\lambda}$-brackets
\begin{align*}
[a \ _{\lambda} \ b] = [a, b],
{\quad}
[a \ _{\lambda} \ {\bar{b}}] = (-1)^{p(a)}{\overline{[a,b]}},
{\quad}
[{\bar{a}} \ _{\lambda} \ {\bar{b}}] = 0,
\end{align*}
for $a, b \in \mathfrak{g}$, and let ${\textup{SVir}} = {\mathbb{C}}[{\partial}] \otimes (\mathbb{C}{L} \oplus \mathbb{C}{G}) \oplus {\mathbb{C}}C$ be the $N = 1$ super-Virasoro Lie conformal algebra with the $\lambda$-brackets in \eqref{eq:superVirasoro relation}, where $C$ is central in the $\mathbb{C}[\partial]$-module ${\textup{SVir}} \oplus {\textup{Cur}\mathfrak{g}}_{N=1}$ and ${\partial}C = 0$. For the ${\lambda}$-brackets given by:
\begin{align*}
[L \ _{\lambda} \ a] = ({\partial} + {\Delta}_{a}{\lambda})a,
{\quad}
[L \ _{\lambda} \ {\bar{a}}] = ({\partial} + ({\Delta}_{a} - \frac{1}{2}){\lambda}){\bar{a}},
\end{align*}
where ${\Delta}_{a}=1-\omega$ for $a \in \mathfrak{g}_{\omega}$, one can reproduce the remaining ${\lambda}$-brackets of the superconformal current algebra as follows
\begin{align*}
[G \ _{\lambda} \ a] = ({\partial} + (2{\Delta}_{a} - 1){\lambda}){\bar{a}},
{\quad}
[G \ _{\lambda} \ {\bar{a}}] = a,
\end{align*}
from the SUSY extension formulas in (SEF1) and (SEF2). Note that ${\Delta}_{[a, b]} = {\Delta}_{a} + {\Delta}_{b} - 1$ for all $a, b \in \mathfrak{g}$, since $\omega_{[a,b]}= \omega_a+ \omega_b$.
Then, one can also check that ${\textup{SVir}} \oplus {\textup{Cur}\mathfrak{g}}_{N=1}$ is indeed a Lie conformal algebra. For example, the following equality holds:
\begin{align*}
[G \ _\lambda \ [a \ _\gamma \ b]] 
&= [G \ _\lambda \ [a, b]] 
= ({\partial} + (2{\Delta}_{[a, b]} - 1){\lambda}){\overline{[a, b]}} 
= ({\partial} + (2{\Delta}_{a} + 2{\Delta}_{b} - 3){\lambda}){\overline{[a, b]}}
\\
&= ((2{\Delta}_{a} - 2){\lambda} - {\gamma}){\overline{[a, b]}} + ({\partial} + (2{\Delta}_{b} - 1){\lambda} + {\gamma}){\overline{[a, b]}}
\\
&= ((2{\Delta}_{a} - 2){\lambda} - {\gamma})[\bar{a} \ _{\lambda + \gamma} \ b] + (-1)^{p(a)}({\partial} + (2{\Delta}_{b} - 1){\lambda} + {\gamma})[a \ _{\gamma} \ \bar{b}]
\\
&= [({\partial} + (2{\Delta}_{a} - 1){\lambda}){\bar{a}} \ _{\lambda + \gamma} \ b] + (-1)^{p(a)}[a \ _{\gamma} \ ({\partial} + (2{\Delta}_{b} - 1){\lambda}){\bar{b}}]
\\
&= [[G \ _{\lambda} \ a] \ _{\lambda + \gamma} \ b] + (-1)^{p(a)}[a \ _{\gamma} \ [G \ _{\lambda} \ b]].
\end{align*}
Hence, $V({\textup{SVir}} \oplus {\textup{Cur}\mathfrak{g}}_{N=1}) / \left< C-c\right> \simeq {\textup{SVir}}^{c} \otimes {V}^{0}_{N=1}(\mathfrak{g})$ is an $N = 1$ superconformal vertex algebra where the odd derivation $D$ is given by $DG = 2L$ and $D{\bar{a}} = {a}$.
\end{remark}

\vskip 1mm

Let $\mathfrak{g}$ be a basic Lie superalgebra. Recall that for $V^{k}_{N=1}(\mathfrak{g})$ with $k\neq 0$, any element $\bar{a} \in \bar{\mathfrak{g}}$ is primary of conformal weight $\frac{1}{2}$ with respect to the Kac-Todorov superconformal vector ${\tau}^{\mathfrak{g}}$. Hence $a:=D\bar{a}\in D(\bar{\mathfrak{g}})$ has the conformal weight $1$. 
In addition, for $h$ in the Cartan subalgebra of $\mathfrak{g}$, the element ${\tau}^{\mathfrak{g}}_{h}$ in \eqref{eq:shifted KT} yields another conformal structure induced from the eigenspace of $\textup{ad} h$.
Similarly, in Remark \ref{remark:superconformal current}, one can construct an $N = 1$ SUSY extension of $V^0(\mathfrak{g})$ such that every element $a\in \mathfrak{g}$ has conformal weight $1$ by taking
$\Omega=\{0\}$. Furthermore, by letting $\Omega$ to be the grading of $\mathfrak{g}$ defined by $\textup{ad} h$, $V^0(\mathfrak{g})$ can be extended to an $N = 1$ superconformal vertex algebra with another conformal weight decomposition. In this case, the conformal weight of $a\in \mathfrak{g}$ is the same as the conformal weight of $a$ in $V^{k}_{N=1}(\mathfrak{g})$ via the superconformal vector ${\tau}^{\mathfrak{g}}_{h}$.

\vskip 6mm

\section{Generalizations}\label{section6}
\setcounter{equation}{0}

The goal of this section is to generalize the discussions in the Section \ref{section5}. Similar to the $N = 1$ case, we say that an $N = n$ SUSY vertex algebra $V$ has an \textit{$N = n' > n$ supersymmetric extension} if there is a differential algebra $W$ such that $V \otimes W$ has an $N = n'$ SUSY vertex algebra structure, containing the original vertex algebra $V$ as an $N = n$ SUSY vertex subalgebra.

\vskip 1mm

Note that we can also impose an $N = 2$ SUSY structure on a vertex algebra $V$, if it has two odd derivations $D^{1}$, $D^{2}$ such that $[D^{i}, D^{j}] = 2{\delta}_{ij}{\partial}$, and the equation \eqref{eq5.1} holds for each $D^{i}$. Then the $\Lambda$-bracket:
\begin{equation}\label{eq:N=2 Lambda bracket}
[a \ _{\Lambda} \ b] = [D^{2}D^{1}a \ _{\lambda} \ b] - {\chi}^{1}[D^{2}a \ _{\lambda} \ b] + {\chi}^{2}[D^{1}a \ _{\lambda} \ b] - {\chi}^{1}{\chi}^{2}[a \ _{\lambda} \ b]
\end{equation}
gives an $N = 2$ SUSY vertex algebra structure on $V$. The proof is similar to the proof of the $N = 1$ case. For example, the non-commutative Wick formula of $N = 2$ SUSY vertex algebras holds since:
\begin{align*}
[a \ _{\Lambda} \ bc] &= [{D}^{2}{D}^{1}a \ _{\lambda} \ bc] - {\chi}^{1}[{D}^{2}a \ _{\lambda} \ bc] + {\chi}^{2}[{D}^{1}a \ _{\lambda} \ bc] - {\chi}^{1}{\chi}^{2}[a \ _{\lambda} \ bc]
\\
&= ([{D}^{2}{D}^{1}a \ _{\lambda} \ b]c + (-1)^{p(a)p(b)}b[{D}^{2}{D}^{1}a \ _{\lambda} \ c] + \int_{0}^{\lambda} [[{D}^{2}{D}^{1}a \ _{\lambda} \ b] \ _{\gamma} \ c] \ d\gamma)
\\
&\quad - {\chi}^{1}([{D}^{2}a \ _{\lambda} \ b]c + (-1)^{(p(a)+1)p(b)}b[{D}^{2}a \ _{\lambda} \ c] + \int_{0}^{\lambda} [[{D}^{2}a \ _{\lambda} \ b] \ _{\gamma} \ c] \ d\gamma)
\\
&\quad + {\chi}^{2}([{D}^{1}a \ _{\lambda} \ b]c + (-1)^{(p(a)+1)p(b)}b[{D}^{1}a \ _{\lambda} \ c] + \int_{0}^{\lambda} [[{D}^{1}a \ _{\lambda} \ b] \ _{\gamma} \ c] \ d\gamma)
\\
&\quad - {\chi}^{1}{\chi}^{2}([a \ _{\lambda} \ b]c + (-1)^{p(a)p(b)}b[a \ _{\lambda} \ c] + \int_{0}^{\lambda} [[a \ _{\lambda} \ b] \ _{\gamma} \ c] \ d\gamma)
\\
&= [a \ _{\Lambda} \ b]c + (-1)^{p(a)p(b)}b[a \ _{\Lambda} \ c]
+ \int_{0}^{\Lambda} [[a \ _{\Lambda} \ b] \ _{\Gamma} \ c] \ d\Gamma.
\end{align*}

\vskip 1mm

Similarly, in the case of $N = n$, we know that $n$ odd derivations of the normally ordered product and the $\lambda$-brackets, denoted as ${D}^{1}$, $\cdots$, ${D}^{n}$, and satisfying $[D^{i}, D^{j}] = 2{\delta}_{ij}{\partial}$, also provide an $N = n$ SUSY vertex algebra structure. For these $n$ odd derivations satisfying the conditions described above, we denote by ${\textup{Der}}_{N = n}(V)$ the set of $n$-tuples of the form $({D}^{1}, \cdots, {D}^{n})$, i.e. the set of all \textit{$N = n$ SUSY structures} of $V$. Then, in the following proposition, we can consider the complex orthogonal group action on the set of $N = n$ SUSY structures of a vertex algebra. Indeed, the orthogonal groups were studied as the automorphism groups of the $N = 2, 3, 4$ superconformal algebras in \cite{SS87}.

\begin{proposition}\label{proposition: orthogonal group action}
Let $V$ be a vertex algebra. Then ${\textup{Der}}_{N = n}(V)$ admits an action of the orthogonal group in dimension $n$.
\end{proposition}

\begin{proof}
Let ${\textbf{O}}_{n}(\mathbb{C})$ be the complex orthogonal group in dimension $n$. If we define an ${\textbf{O}}_{n}(\mathbb{C})$-action by
\begin{align*}
A \cdot ({D}^{1}, \cdots, {D}^{n}) := (\sum_{i = 1}^{n} A_{1i}{D}^{i}, \cdots, \sum_{i = 1}^{n} A_{ni}{D}^{i}),
\end{align*}
for $A \in {\textbf{O}}_{n}(\mathbb{C})$ and $({D}^{1}, \cdots, {D}^{n}) \in {\textup{Der}}_{N = n}(V)$, then the result follows from direct calculations.
\end{proof}

\vskip 1mm

Note that a vertex algebra is called an {\it $N = 2$ superconformal vertex algebra} if it has two odd vectors ${G}^{\pm}$ satisfying the $N = 2$ super-Virasoro relation \eqref{eq:N=2 superVirasoro relation} with a conformal vector $L = {\frac{1}{2}}({{G}^{+}}_{(0)}{G}^{-} + {{G}^{-}}_{(0)}{G}^{+})$ and $J = {{G}^{+}}_{(1)}{G}^{-}$. Then, from Remark 2.8 of \cite{HK07} (see also Remark 1.3 of \cite{Barron10}), an $N = 2$ superconformal vertex algebra has an $N = 2$ SUSY vertex algebra structure by taking ${D}^{1} = ({\mu}{G}^{+} + {\mu}^{-1}{G}^{-})_{(0)}$ and ${D}^{2} = {\pm}{\sqrt{-1}}({\mu}{G}^{+} - {\mu}^{-1}{G}^{-})_{(0)}$ for ${\mu} \in {\mathbb{C}}^{*}$. This ${\mathbb{Z}} \slash {2\mathbb{Z}} \cross {\mathbb{C}}^{*} \simeq {\mathbb{Z}} \slash {2\mathbb{Z}} \cross {\textbf{SO}}_{2}(\mathbb{C})$ family of $N = 2$ SUSY structures is derived from the automorphism group of the $N = 2$ super-Virasoro vertex algebra, which preserve the conformal vector. Therefore, Proposition \ref{proposition: orthogonal group action} can be viewed as a generalization of the relationships between SUSY vertex algebra structures induced from the superconformal structures.

\vskip 1mm

We will now generalize the result from Section \ref{section5}. Recall that we have shown that there exists an $N = 1$ SUSY extension of any universal enveloping vertex algebras. Along with the result of Theorem \ref{theorem:N=1 extension}, we can also consider the following question, which will be proven in Theorem \ref{theorem:N=2 extension}.

\begin{question} \label{question:N=1 to N=2}
\rm
For a given $N = 1$ supersymmetric Lie conformal algebra ${R}_{N = 1}$, can we find an $N = 2$ supersymmetric extension of the $N = 1$ supersymmetric vertex algebra $V({R}_{N = 1})$?
\end{question}

\vskip 1mm

In order to answer to Question \ref{question:N=1 to N=2}, we aim to find conditions to assign an $N=2$ supersymmetry on the ${\mathbb{C}}[{\partial}]$-module:
\begin{equation}\label{eq:sec6, C[partial]-module}
\widetilde{R} = {\mathbb{C}}[{\partial}] \otimes {\textup{span}}_{\mathbb{C}}\{ {u}_{i}, {\bar{u}}_{i}, {u}^{\circ}_{i}, {\bar{u}}^{\circ}_{i} \}_{ i \in I } \oplus E_{C},  
\end{equation}
where the set $\{ {u}_{i}, {\bar{u}}_{i}, {u}^{\circ}_{i}, {\bar{u}}^{\circ}_{i} \}_{ i \in I }$ is linearly independent and $E_C$ is the central extension part of $\widetilde{R}$ such that $\partial E_C=0$. Additionally, we assume the parity condition in (D-4) and the existence of two odd endomorphisms with properties in (D-5).

\begin{enumerate}[]
\item (D-4) \ For the generators of $\widetilde{R}$, their parities satisfy $p({u}_{i}) = p({\bar{u}}^{\circ}_{i}) = 1 - p({\bar{u}}_{i}) = 1 - p({u}^{\circ}_{i})$ for any $i \ {\in} \ I$.
\item (D-5) \ ${D}^{1}$, ${D}^{2}$ are odd endomorphisms on $\widetilde{R}$ such that $[D^{i}, D^{j}] = 2{\delta}_{ij}{\partial}$, and satisfy the following diagram:
\[
\begin{tikzcd}[style={font=\small}]
& {\bar{u}}_{i} \arrow{rr}{{D}^{1}} & & {u}_{i} \\
{\bar{u}}^{\circ}_{i} \arrow{rr}[swap]{-{D}^{1}} \arrow{ru}{{D}^{2}} & & {u}^{\circ}_{i} \arrow{ru}[swap]{{D}^{2}}
\end{tikzcd}
\]
\noindent for each $i \in I$, i.e. ${D}^{1}({\bar{u}}_{i}) = {u}_{i}$, ${D}^{1}({\bar{u}}^{\circ}_{i}) = -{u}^{\circ}_{i}$, ${D}^{2}({u}^{\circ}_{i}) = {u}_{i}$, ${D}^{2}({\bar{u}}^{\circ}_{i}) = {\bar{u}}_{i}$, and $D^1(E_{C})=D^2(E_{C}) = 0$.
\end{enumerate}

\begin{proposition}\label{proposition:N=2 extension strongly generated}
Let ${\widetilde{R}}$ be the ${\mathbb{C}}[{\partial}]$-module in \eqref{eq:sec6, C[partial]-module} satisfying ${\textup{(D-4)}}$, and let ${D}^{1}$ and ${D}^{2}$ be given by the data in ${\textup{(D-5)}}$. Suppose the ${\lambda}$-brackets between the elements in $\{{u}_{i}, {\bar{u}}_{i}, {u}^{\circ}_{i}, {\bar{u}}^{\circ}_{i}\}_{i {\in} I}$ are given so that ${\widetilde{R}}$ is a Lie conformal algebra, and each of the following pairs
\begin{align}\label{eq:N=2 SUSY generator}
(\{{\bar{u}}_{i}, {\bar{u}}^{\circ}_{i}\}_{i {\in} I}, {D}^{1}),
{\quad}
(\{{u}^{\circ}_{i}, {\bar{u}}^{\circ}_{i}\}_{i {\in} I}, {D}^{2}),
\end{align}
satisfies the SUSY extension formulas. Then $V(\widetilde{R})$ is an $N = 2$ SUSY vertex algebra with the odd derivations ${D}^{1}$ and ${D}^{2}$.
\end{proposition}

\begin{proof}
By the proof of Proposition \ref{proposition:N=1 extension strongly generated}, we know that each of ${D}^{1}$ and ${D}^{2}$ acts as an odd derivation of the ${\lambda}$-brackets on $V({\widetilde{R}})$. Hence the universal enveloping vertex algebra $V({\widetilde{R}})$ is an $N = 2$ SUSY vertex algebra via the ${\Lambda}$-brackets in \eqref{eq:N=2 Lambda bracket}.
\end{proof}

\vskip 1mm

Note that the SUSY extension formulas for the pair $(\{{\bar{u}}_{i}, {\bar{u}}^{\circ}_{i}\}_{i {\in} I}, {D}^{1})$ in \eqref{eq:N=2 SUSY generator} can be written explicitly as the following equalities:
\begin{align*}
[{u}^{\circ}_{i} \ _{\lambda} \ {u}^{\circ}_{j}] &= (-1)^{p({u}_{i})}(D^{1}[{u}^{\circ}_{i} \ _{\lambda} \ {\bar{u}}^{\circ}_{j}] - {\lambda}[{\bar{u}}^{\circ}_{i} \ _{\lambda} \ {\bar{u}}^{\circ}_{j}]),
&\quad
[{\bar{u}}^{\circ}_{i} \ _{\lambda} \ {u}^{\circ}_{j}] &= (-1)^{p({u}_{i})+1}([{u}^{\circ}_{i} \ _{\lambda} \ {\bar{u}}^{\circ}_{j}] + D^{1}[{\bar{u}}^{\circ}_{i} \ _{\lambda} \ {\bar{u}}^{\circ}_{j}]),
\\
[{u}_{i} \ _{\lambda} \ {u}^{\circ}_{j}] &= (-1)^{p({u}_{i})+1}(D^{1}[{u}_{i} \ _{\lambda} \ {\bar{u}}^{\circ}_{j}] + {\lambda}[{\bar{u}}_{i} \ _{\lambda} \ {\bar{u}}^{\circ}_{j}]),
&\quad
[{\bar{u}}_{i} \ _{\lambda} \ {u}^{\circ}_{j}] &= (-1)^{p({u}_{i})+1}([{u}_{i} \ _{\lambda} \ {\bar{u}}^{\circ}_{j}] - D^{1}[{\bar{u}}_{i} \ _{\lambda} \ {\bar{u}}^{\circ}_{j}]),
\\
[{u}_{i} \ _{\lambda} \ {u}_{j}] &= (-1)^{p({u}_{i})}({D}^{1}[{u}_{i} \ _{\lambda} \ {\bar{u}}_{j}] + {\lambda}[{\bar{u}}_{i} \ _{\lambda} \ {\bar{u}}_{j}]),
&\quad
[{\bar{u}}_{i} \ _{\lambda} \ {u}_{j}] &= (-1)^{p({u}_{i})}([{u}_{i} \ _{\lambda} \ {\bar{u}}_{j}] - {D}^{1}[{\bar{u}}_{i} \ _{\lambda} \ {\bar{u}}_{j}]),
\end{align*}
and the SUSY extension formulas for the second pair $(\{{u}^{\circ}_{i}, {\bar{u}}^{\circ}_{i}\}_{i {\in} I}, {D}^{2})$ are given by the following equalities:
\begin{align*}
[{\bar{u}}_{i} \ _{\lambda} \ {\bar{u}}_{j}] &= (-1)^{p({u}_{i})+1}(D^{2}[{\bar{u}}_{i} \ _{\lambda} \ {\bar{u}}^{\circ}_{j}] + {\lambda}[{\bar{u}}^{\circ}_{i} \ _{\lambda} \ {\bar{u}}^{\circ}_{j}]),
&\quad
[{\bar{u}}^{\circ}_{i} \ _{\lambda} \ {\bar{u}}_{j}] &= (-1)^{p({u}_{i})+1}([{\bar{u}}_{i} \ _{\lambda} \ {\bar{u}}^{\circ}_{j}] - D^{2}[{\bar{u}}^{\circ}_{i} \ _{\lambda} \ {\bar{u}}^{\circ}_{j}]),
\\
[{u}_{i} \ _{\lambda} \ {\bar{u}}_{j}] &= (-1)^{p({u}_{i})}(D^{2}[{u}_{i} \ _{\lambda} \ {\bar{u}}^{\circ}_{j}] + {\lambda}[{u}^{\circ}_{i} \ _{\lambda} \ {\bar{u}}^{\circ}_{j}]),
&\quad
[{u}^{\circ}_{i} \ _{\lambda} \ {\bar{u}}_{j}] &= (-1)^{p({u}_{i})}([{u}_{i} \ _{\lambda} \ {\bar{u}}^{\circ}_{j}] - D^{2}[{u}^{\circ}_{i} \ _{\lambda} \ {\bar{u}}^{\circ}_{j}]),
\\
[{u}_{i} \ _{\lambda} \ {u}_{j}] &= (-1)^{p({u}_{i})}({D}^{2}[{u}_{i} \ _{\lambda} \ {u}^{\circ}_{j}] + {\lambda}[{u}^{\circ}_{i} \ _{\lambda} \ {u}^{\circ}_{j}]),
&\quad
[{u}^{\circ}_{i} \ _{\lambda} \ {u}_{j}] &= (-1)^{p({u}_{i})}([{u}_{i} \ _{\lambda} \ {u}^{\circ}_{j}] - {D}^{2}[{u}^{\circ}_{i} \ _{\lambda} \ {u}^{\circ}_{j}]).
\end{align*}

\vskip 1mm

Proposition \ref{proposition:N=2 extension strongly generated} tells that, for a given Lie conformal algebra $R$, if we can find a Lie conformal algebra $\widetilde{R}$ that contains $R$ as a Lie conformal subalgebra and satisfies the conditions of the proposition, then we can conclude that $V(R)$ has an $N = 2$ SUSY extension.

\vskip 1mm

\begin{example}[$N=2$ SUSY structure of $N=2$ super-Virasoro vertex algebra]\label{example:N=2 SUSY of N=2 superVirasoro}
\rm From Example 2.9 of \cite{HK07}, we already know that the $N = 2$ super-Virasoro vertex algebra (see also Examples \ref{example:superVirasoro VA} and \ref{example:SUSY superVirasoro VA}) has an $N = 2$ SUSY vertex algebra structure as in the diagram:
\[
\begin{tikzcd}[row sep=scriptsize, column sep=scriptsize, style={font=\small}]
& {\sqrt{-1}}({G}^{+} + {G}^{-}) \arrow[rr, "{D}^{1}"] & & 2{\sqrt{-1}}{L}
\\
{J} \arrow[rr, swap, "-{D}^{1}"] \arrow[ru, "{D}^{2}"] & & {G}^{+} - {G}^{-} \arrow[ru, swap, "{D}^{2}"]
\end{tikzcd}
\]
Indeed, if we choose a generating set of the $N = 2$ super-Virasoro vertex algebra $V(\widetilde{R})/\left< C - c \right>$ as
\begin{align*}
{\bar{u}}^{\circ}_{1} = {J}, \quad {u}^{\circ}_{1} = {G}^{+} - {G}^{-}, \quad {\bar{u}}_{1} = {\sqrt{-1}}({G}^{+} + {G}^{-}), \quad {u}_{1} = 2{\sqrt{-1}}{L},
\end{align*}
where $\widetilde{R} = {\mathbb{C}}[{\partial}] \otimes {\textup{span}}_{\mathbb{C}}\{ {u}_{1}, {\bar{u}}_{1}, {u}^{\circ}_{1}, {\bar{u}}^{\circ}_{1} \} \oplus {\mathbb{C}}C$, then  the SUSY extension formulas for the pairs in (\ref{eq:N=2 SUSY generator}) can be checked by direct computations.
\end{example}

\vskip 1mm

In the consecutive examples (Example \ref{ex:N=2, extended bcbetagamma}, \ref{ex:N=2, extended bcbetagamma-2} and \ref{ex:N=2, extended bcbetagamma-3}), we investigate several $N = 2$ SUSY vertex algebra structures of the $N = 2$ SUSY extension of the $bc$–${\beta}{\gamma}$ system.
The subsequent example is a construction of an $N = 2$ SUSY extension of the $bc$–${\beta}{\gamma}$ system via Proposition \ref{proposition:N=2 extension strongly generated}. 

\begin{example}[$N = 2$ SUSY extension of $bc$–${\beta}{\gamma}$ system] \label{ex:N=2, extended bcbetagamma}
\rm
Let ${\widetilde{R}} = {\mathbb{C}}[{\partial}] \otimes {\textup{span}}_{\mathbb{C}}\{ {u}_{i}, {\bar{u}}_{i}, {u}^{\circ}_{i}, {\bar{u}}^{\circ}_{i} \}_{ i \in \{ 1, 2 \} } \oplus {\mathbb{C}}C$ be a ${\mathbb{C}[{\partial}]}$-module where $C$ is a central element satisfying ${\partial}C = 0$ and the generators are given by
\[
{u}_{1} = {\beta},
{\quad}
{\bar{u}}_{1} = b,
{\quad}
{u}^{\circ}_{1} = -a,
{\quad}
{\bar{u}}^{\circ}_{1} = {\alpha},
{\quad}
{u}_{2} = {\delta},
{\quad}
{\bar{u}}_{2} = d,
{\quad}
{u}^{\circ}_{2} = -c,
{\quad}
{\bar{u}}^{\circ}_{2} = {\gamma},
\]
where $\alpha$, $\beta$, $\gamma$, $\delta$ are even elements and $a$, $b$, $c$, $d$ are odd elements.
Define the odd endomorphisms ${D}^{1}$ and ${D}^{2}$ on $\widetilde{R}$ as in the diagram below
\[
\begin{tikzcd}[row sep=scriptsize, column sep=scriptsize, style={font=\small}]
& {b} \arrow{rr}{{D}^{1}} & & {\beta} &
& {d} \arrow{rr}{{D}^{1}} & & {\delta}
\\
{\alpha} \arrow{rr}[swap]{-{D}^{1}} \arrow{ru}{{D}^{2}} & & {-a} \arrow{ru}[swap]{{D}^{2}} & &
{\gamma} \arrow{rr}[swap]{-{D}^{1}} \arrow{ru}{{D}^{2}} & & {-c} \arrow{ru}[swap]{{D}^{2}}
\end{tikzcd}
\]
\noindent with $[D^{i}, D^{j}] = 2{\delta}_{ij}{\partial}$ and ${D}^{1}C = {D}^{2}C = 0$.
For the Lie conformal algebra structure of $\widetilde{R}$, define the non-zero $\lambda$-brackets of $\widetilde{R}$ as
\begin{gather*}
[{\beta} \ _{\lambda} \ {\gamma}] = C, \quad [{\gamma} \ _{\lambda} \ {\beta}] = -C, \quad [b \ _{\lambda} \ c] = C, \quad [c \ _{\lambda} \ b] = C,
\\
[{\delta} \ _{\lambda} \ {\alpha}] = -C, \quad [{\alpha} \ _{\lambda} \ {\delta}] = C, \quad [d \ _{\lambda} \ a] = -C, \quad [a \ _{\lambda} \ d] = -C.
\end{gather*}
Then, we can check each of the pairs $(\{{\bar{u}}_{i}, {\bar{u}}^{\circ}_{i}\}_{i \in \{ 1, 2 \}}, {D}^{1})$ and $(\{{u}^{\circ}_{i}, {\bar{u}}^{\circ}_{i}\}_{i \in \{ 1, 2 \}}, {D}^{2})$ satisfies the SUSY extension formulas.
Hence, by Proposition \ref{proposition:N=2 extension strongly generated}, $V(\widetilde{R})$ is an $N = 2$ SUSY extension of the $bc$–${\beta}{\gamma}$ system where the non-zero ${\Lambda}$-bracket of the generators as an $N = 2$ SUSY vertex algebra is
\begin{align*}
[{\alpha} \ _{\Lambda} \ {\gamma}] = -C,
\end{align*}
and, we call the vertex algebra $V(\widetilde{R})$ or its quotient vertex algebra $V(\widetilde{R})/\left< C - 1 \right>$ the {\it extended $bc$–${\beta}{\gamma}$ system} in the rest of the paper.
\end{example}

In the following examples, we describe other $N = 2$ SUSY structures derived from $N = 2$ superconformal structures in the extended $bc$–${\beta}{\gamma}$ system (see also \cites{EHKZ09, Hel09}).

\begin{example}[$N = 2$ superconformal structures of the extended $bc$–${\beta}{\gamma}$ system (1)]\label{ex:N=2, extended bcbetagamma-2}
\rm
The extended $bc$–${\beta}{\gamma}$ system has an $N = 2$ superconformal structure associated to the following vectors
\begin{gather*}
{L}_{h} = {\frac{1}{2}}\big( - {c}{\partial}{b} + {\partial}{c}{b} + 2{\partial}{\gamma}{\beta} + {a}{\partial}{d} - {\partial}{a}{d} - 2{\partial}{\alpha}{\delta} \big),
\\
{G}^{+}_{h} = {\sqrt{-1}}({c}{\beta} - {a}{\delta}),
{\quad}
{G}^{-}_{h} = {\sqrt{-1}}(- {\partial}{\gamma}{b} + {\partial}{\alpha}{d}),
{\quad}
{J}_{h} = {c}{b} - {a}{d},
\end{gather*}
with central charge $6$, and the $N=2$ SUSY vertex algebra structure induced from the above $N = 2$ superconformal structure is depicted in the following picture
\[
\begin{tikzcd}[row sep=scriptsize, column sep=scriptsize, style={font=\small}]
& {-{\beta}} \arrow{rr}{{D}^{1}_{h}} & & {{\sqrt{-1}}{\partial}b} & & {-c} \arrow{rr}{{D}^{1}_{h}} & & {{\sqrt{-1}}{\partial}{\gamma}}
\\
{b} \arrow{rr}[swap]{-{D}^{1}_{h}} \arrow{ru}{{D}^{2}_{h}} & & {-{\sqrt{-1}}{\beta}} \arrow{ru}[swap]{{D}^{2}_{h}} & &
{\gamma} \arrow{rr}[swap]{-{D}^{1}_{h}} \arrow{ru}{{D}^{2}_{h}} & & {-{\sqrt{-1}}c} \arrow{ru}[swap]{{D}^{2}_{h}} & 
\\
& {-a} \arrow{rr}{{D}^{1}_{h}} & & {{\sqrt{-1}}{\partial}{\alpha}} &
& {-{\delta}} \arrow{rr}{{D}^{1}_{h}} & & {{\sqrt{-1}}{\partial}d}
\\
{\alpha} \arrow{rr}[swap]{-{D}^{1}_{h}} \arrow{ru}{{D}^{2}_{h}} & & {-{\sqrt{-1}}a} \arrow{ru}[swap]{{D}^{2}_{h}} & &
{d} \arrow{rr}[swap]{-{D}^{1}_{h}} \arrow{ru}{{D}^{2}_{h}} & & {-{\sqrt{-1}}{\delta}} \arrow{ru}[swap]{{D}^{2}_{h}}
\end{tikzcd}
\]
where ${D}^{1}_{h} = ({G}^{+}_{h} + {G}^{-}_{h})_{(0)}$ and ${D}^{2}_{h} = {\sqrt{-1}}({G}^{+}_{h} - {G}^{-}_{h})_{(0)}$. 
Observe that the non-zero ${\Lambda}$-brackets of generators are given by
\begin{align*}
[{b} \ _{\Lambda} \ {\gamma}] = {\chi}^{1} + {\sqrt{-1}}{\chi}^{2},
{\quad}
[{\alpha} \ _{\Lambda} \ {d}] = - {\chi}^{1} - {\sqrt{-1}}{\chi}^{2},
\end{align*}
and hence each of the $bc$–$\beta\gamma$ system and {\it the $ad$–$\alpha\delta$ system} (which refers to the vertex algebra generated by the four elements $a$, $d$, $\alpha$, $\delta$) is already closed under the $N=2$ SUSY vertex algebra structure, respectively. In other words, this $N = 2$ SUSY vertex algebra structure which is derived from the $N = 2$ superconformal structure associated to ${G}^{\pm}_{h}$ provides an {\it internal $N = 2$ SUSY structure} for both the $bc$–$\beta\gamma$ system and the $ad$–$\alpha\delta$ system.
Therefore, the extended $bc$–${\beta}{\gamma}$ system with  ${D}^{1}_{h}$ and ${D}^{2}_{h}$ can be understood as the product of two independent copies of the $bc$–${\beta}{\gamma}$ system as $N=2$ SUSY vertex algebras.
\end{example}

\begin{example}[$N = 2$ superconformal structures of the extended $bc$–${\beta}{\gamma}$ system (2)] \label{ex:N=2, extended bcbetagamma-3}
\rm
 We have another $N = 2$ superconformal structure of the extended $bc$–${\beta}{\gamma}$ system induced from the following vectors
\begin{gather*}
{L}_{v} = {L}_{h},
{\quad}
{G}^{+}_{v} = {d}{\beta} - {b}{\delta},
{\quad}
{G}^{-}_{v} = - {\partial}{\gamma}{a} + {\partial}{\alpha}{c},
{\quad}
{J}_{v} = -{c}{b} + {a}{d}.
\end{gather*}
The $N = 2$ SUSY structure associated to ${G}^{\pm}_{v}$ is presented in the diagram below:
\[
\begin{tikzcd}[row sep=scriptsize, column sep=scriptsize, style={font=\small}]
& {{\sqrt{-1}}b} \arrow{rr}{{D}^{1}_{v}} & & {{\sqrt{-1}}{\partial}{\alpha}} &
& {-{\sqrt{-1}}{\delta}} \arrow{rr}{{D}^{1}_{v}} & & {{\sqrt{-1}}{\partial}c}
\\ 
{\alpha} \arrow{rr}[swap]{-{D}^{1}_{v}} \arrow{ru}{{D}^{2}_{v}} & & {-b} \arrow{ru}{{D}^{2}_{v}} & &
{c} \arrow{rr}[swap]{-{D}^{1}_{v}} \arrow{ru}{{D}^{2}_{v}} & & {\delta} \arrow{ru}{{D}^{2}_{v}} & 
\\
& {-{\sqrt{-1}}{\beta}} \arrow{rr}{{D}^{1}_{v}} & & {{\sqrt{-1}}{\partial}{a}} & & {{\sqrt{-1}}d} \arrow{rr}{{D}^{1}_{v}} & & {{\sqrt{-1}}{\partial}{\gamma}}
\\
{a} \arrow{rr}[swap]{-{D}^{1}_{v}} \arrow{ru}{{D}^{2}_{v}} & & {\beta} \arrow{ru}{{D}^{2}_{v}} & &
{\gamma} \arrow{rr}[swap]{-{D}^{1}} \arrow{ru}{{D}^{2}} & & {-d} \arrow{ru}{{D}^{2}} &
\end{tikzcd}
\]
and the non-zero ${\Lambda}$-brackets of generators are
\begin{align*}
[{\alpha} \ _{\Lambda} \ {c}] = -{\sqrt{-1}}{\chi}^{1} + {\chi}^{2},
{\quad}
[{a} \ _{\Lambda} \ {\gamma}] = {\sqrt{-1}}{\chi}^{1} - {\chi}^{2}.
\end{align*}
Hence the extended $bc$–${\beta}{\gamma}$ system with ${D}^{1}_{v}$ and ${D}^{2}_{v}$ is the product of the $bc$–${\alpha}{\delta}$ system and the $ad$–$\beta\gamma$ system, each having its own internal $N = 2$ SUSY structure.
\end{example}

\vskip 1mm

We remark that each of the odd derivations ${D}^{1}$, ${D}^{2}$ in Example \ref{ex:N=2, extended bcbetagamma} can be realized as an odd derivation derived from the independent superconformal structures, with ${D}^{1} = -{\sqrt{-1}}{D}^{1}_{h} = -{D}^{2}_{h}$ for $D^1_{h}$ and $D^2_{h}$ in Example \ref{ex:N=2, extended bcbetagamma-2}, and ${D}^{2} = {D}^{1}_{v} = -{\sqrt{-1}}{D}^{2}_{v}$ for $D^1_{v}$ and $D^2_{v}$ in Example \ref{ex:N=2, extended bcbetagamma-3}.
Furthermore, there is an $N = 2$ superconformal structure of the extended $bc$–${\beta}{\gamma}$ system that simultaneously induces ${D}^{1}$ and ${D}^{2}$ (see Section 6.4 of \cite{EHKZ09}). Precisely, the $N = 2$ superconformal structure induced from the following vectors 
\begin{gather*}
{L} = {L}_{h},
{\quad}
{G}^{+} = -{\frac{{\sqrt{-1}}}{2}}({G}^{+}_{h} - {G}^{-}_{h} + {G}^{+}_{v} + {G}^{-}_{v}),
\\
{G}^{-} = -{\frac{{\sqrt{-1}}}{2}}({G}^{+}_{h} - {G}^{-}_{h} - {G}^{+}_{v} - {G}^{-}_{v}),
{\quad}
{J} = {\sqrt{-1}}({a}{c} + {b}{d}),
\end{gather*}
give rise to the $N = 2$ SUSY vertex algebra structure of the extended $bc$–${\beta}{\gamma}$ system defined in Example \ref{ex:N=2, extended bcbetagamma}.

\vskip 1mm

Now, returning to the Question \ref{question:N=1 to N=2}, let us construct an $N = 2$ SUSY extension of a universal enveloping $N = 1$ SUSY vertex algebra $V(R_{N = 1})$. For the purpose of finding $N = 2$ SUSY extensions, we use the construction in Theorem \ref{theorem:N=1 extension} and give an additional $N = 1$ supersymmetry in an independent direction from the given $N = 1$ supersymmetry of $V(R_{N = 1})$.

\begin{theorem}\label{theorem:N=2 extension}
Let ${R}_{N = 1} = {\mathbb{C}}[{D}] \otimes {\textup{span}}_{\mathbb{C}}\{ {\bar{u}}_{i} \}_{ i {\in} I } \oplus E_{C}$ be an $N = 1$ SUSY Lie conformal algebra, where the set $\{ D({\bar{u}}_{i}), {\bar{u}}_{i} \}_{ i {\in} I }$ is linearly independent and $E_C$ is the central extension part of ${R}_{N = 1}$ with ${D}{E}_{C} = 0$. Then $V({R}_{N = 1})$ has an $N = 2$ SUSY extension.
\end{theorem}

\begin{proof}
If we set ${u}_{i} := {D}{\bar{u}}_{i}$, then we can consider the $N = 1$ SUSY Lie conformal algebra ${R}_{N = 1}$ as a non-SUSY Lie conformal algebra ${R}_{N = 1} = {\mathbb{C}}[{\partial}] \otimes {\textup{span}}_{\mathbb{C}}\{ {u}_{i}, {\bar{u}}_{i} \}_{ i {\in} I } \oplus E_{C}$ with a SUSY generator $\{{\bar{u}}_{i}\}_{i {\in} I}$ in the $D$–direction.

By using the construction in Theorem \ref{theorem:N=1 extension}, we have a SUSY generator $\{ {u}^{\circ}_{i}, {\bar{u}}^{\circ}_{i} \}_{ i {\in} I }$ in the ${D}^{\circ}$–direction of the ${\mathbb{C}}[{\partial}]$–module ${\widetilde{R}}_{N = 1} := {{R}_{N = 1}} \oplus {R}_{N = 1}^{\, \circ}$, where ${{R}}_{N = 1}^{\, \circ} = {\mathbb{C}}[{\partial}] \otimes {\textup{span}}_{\mathbb{C}}\{ {u}^{\circ}_{i}, {\bar{u}}^{\circ}_{i} \}_{ i {\in} I } \oplus {E}^{\, \circ}_{C}$.
The remaining non-zero ${\lambda}$-brackets on ${\widetilde{R}}_{N = 1}$ are given by the following relations
\begin{align*}
[{u}^{\circ}_{i} \ _{\lambda} \ {u}_{j}] &= {[{u}_{i} \ _{\lambda} \ {u}_{j}]}^{\circ},
\quad
[{u}_{i} \ _{\lambda} \ {u}^{\circ}_{j}] = (-1)^{p({u}_{i})}{[{u}_{i} \ _{\lambda} \ {u}_{j}]}^{\circ},
\\
[{u}^{\circ}_{i} \ _{\lambda} \ {\bar{u}}_{j}] &= {[{u}_{i} \ _{\lambda} \ {\bar{u}}_{j}]}^{\circ},
\quad
[{u}_{i} \ _{\lambda} \ {\bar{u}}^{\circ}_{j}] = (-1)^{p({u}_{i})}{[{u}_{i} \ _{\lambda} \ {\bar{u}}_{j}]}^{\circ},
\\
[{\bar{u}}^{\circ}_{i} \ _{\lambda} \ {u}_{j}] &= {[{\bar{u}}_{i} \ _{\lambda} \ {u}_{j}]}^{\circ},
\quad
[{\bar{u}}_{i} \ _{\lambda} \ {u}^{\circ}_{j}] = (-1)^{p({u}_{i})+1}{[{\bar{u}}_{i} \ _{\lambda} \ {u}_{j}]}^{\circ},
\\
[{\bar{u}}^{\circ}_{i} \ _{\lambda} \ {\bar{u}}_{j}] &= {[{\bar{u}}_{i} \ _{\lambda} \ {\bar{u}}_{j}]}^{\circ},
\quad
[{\bar{u}}_{i} \ _{\lambda} \ {\bar{u}}^{\circ}_{j}] = (-1)^{p({u}_{i})+1}{[{\bar{u}}_{i} \ _{\lambda} \ {\bar{u}}_{j}]}^{\circ}.
\end{align*}
Now, to give an $N = 2$ SUSY structure on ${\widetilde{R}}_{N = 1}$, let us define odd derivations ${D}^{1}$ and ${D}^{2}$ as follows:
\begin{gather*}
{D}^{1}({{u}_{i}}) := {D}({u}_{i}),
{\quad}
{D}^{1}({\bar{u}}_{i}) := {D}({\bar{u}}_{i}),
{\quad}
{D}^{1}({{u}^{\circ}_{i}}) := (-D({u}_{i}))^{\circ},
{\quad}
{D}^{1}({\bar{u}}^{\circ}_{i}) := (-D({\bar{u}}_{i}))^{\circ},
{\quad}
{D}^{2} := {D}^{\circ}.
\end{gather*}
Then it can be verified that the pair $(\{{\bar{u}}_{i}, {\bar{u}}^{\circ}_{i}\}_{i {\in} I}, {D}^{1})$ also satisfies the SUSY extension formulas, especially the non-trivial parts can be verified through the following calculations:
\begin{align*}
[{u}_{i} \ _{\lambda} \ {u}^{\circ}_{j}] 
&= (-1)^{p({u}_{i})}[{u}_{i} \ _{\lambda} \ {u}_{j}]^{\circ} 
= ({D}[{u}_{i} \ _{\lambda} \ {\bar{u}}_{j}] + {\lambda}[{\bar{u}}_{i} \ _{\lambda} \ {\bar{u}}_{j}])^{\circ}
\\
&= ({D}[{u}_{i} \ _{\lambda} \ {\bar{u}}_{j}])^{\circ} + {\lambda}[{\bar{u}}_{i} \ _{\lambda} \ {\bar{u}}_{j}]^{\circ} 
= - {D}^{1}([{u}_{i} \ _{\lambda} \ {\bar{u}}_{j}]^{\circ}) + {\lambda}[{\bar{u}}_{i} \ _{\lambda} \ {\bar{u}}_{j}]^{\circ}
\\
&= (-1)^{p({u}_{i})+1}(D^{1}[{u}_{i} \ _{\lambda} \ {\bar{u}}^{\circ}_{j}] + {\lambda}[{\bar{u}}_{i} \ _{\lambda} \ {\bar{u}}^{\circ}_{j}]),
\\
[{\bar{u}}_{i} \ _{\lambda} \ {u}^{\circ}_{j}] 
&= (-1)^{p({u}_{i}) + 1}[{\bar{u}}_{i} \ _{\lambda} \ {u}_{j}]^{\circ} 
= - ([{u}_{i} \ _{\lambda} \ {\bar{u}}_{j}] - {D}[{\bar{u}}_{i} \ _{\lambda} \ {\bar{u}}_{j}])^{\circ}
\\
&= - [{u}_{i} \ _{\lambda} \ {\bar{u}}_{j}]^{\circ} + ({D}[{\bar{u}}_{i} \ _{\lambda} \ {\bar{u}}_{j}])^{\circ} 
= - [{u}_{i} \ _{\lambda} \ {\bar{u}}_{j}]^{\circ} - {D}^{1}([{\bar{u}}_{i} \ _{\lambda} \ {\bar{u}}_{j}]^{\circ})
\\
&= (-1)^{p({u}_{i})+1}([{u}_{i} \ _{\lambda} \ {\bar{u}}^{\circ}_{j}] - D^{1}[{\bar{u}}_{i} \ _{\lambda} \ {\bar{u}}^{\circ}_{j}]),
\end{align*}
which follow from the fact that $\{{\bar{u}}_{i}\}_{i {\in} I}$ is a SUSY generator in the $D$–direction. Therefore, by Proposition \ref{proposition:N=2 extension strongly generated}, $V({\widetilde{R}}_{N = 1}) = V({R}_{N = 1}) \otimes V({R}_{N = 1}^{\, \circ})$ is an $N = 2$ SUSY vertex algebra.
\end{proof}

\vskip 1mm

Combining the results in Section \ref{section5} and Theorem \ref{theorem:N=2 extension}, we conclude that any non-SUSY universal enveloping vertex algebra can be embedded into an $N = 2$ SUSY vertex algebra.

\begin{corollary}\label{corollary:N=2 extension}
Let ${R} = {\mathbb{C}}[{\partial}] \otimes {\textup{span}}_{\mathbb{C}}\{ {u}_{i} \}_{ i {\in} I } \oplus E_{C}$ be a Lie conformal algebra in \eqref{eq:sec 5,LCA}. Then $V(R)$ has an $N = 2$ SUSY extension.
\end{corollary}

\begin{proof}
It follows from Theorem \ref{theorem:N=1 extension} and \ref{theorem:N=2 extension}.
\end{proof}

\vskip 1mm

Recall that the $N = 2$ super-Virasoro vertex algebra is an $N=2$ SUSY extension of 
the $N = 1$ super-Virasoro vertex algebra, which is an $N=1$ SUSY extension of the Virasoro vertex algebra. In the following examples, 
by applying Theorem \ref{theorem:N=2 extension} (resp. Corollary \ref{corollary:N=2 extension}) to the $N = 1$ super-Virasoro (resp. Virasoro) vertex algebra, we obtain $N = 2$ SUSY extensions of the Virasoro vertex algebra, which are distinct from the $N = 2$ super-Virasoro vertex algebra.

\begin{example}[$N = 2$ SUSY extension of Virasoro and $N = 1$ super-Virasoro vertex algebra] \label{N=2 SUSY extension of Virasoro-1}
\rm
In this example, we consider the centerless Virasoro and $N = 1$ super-Virasoro Lie conformal algebras (see Example \ref{ex:N=1 ext of vir}).
Recall that the $N = 1$ super-Virasoro Lie conformal algebra $\textup{SVir} = \mathbb{C}[\partial]\otimes (\mathbb{C}L \oplus \mathbb{C}G)$ admits an $N = 1$ supersymmetry defined by $DG = 2L$. If we apply the construction of Theorem \ref{theorem:N=2 extension} to the $N = 1$ super-Virasoro vertex algebra, we obtain the $N = 2$ SUSY vertex algebra $V(\widetilde{\textup{SVir}})$, where $\widetilde{\textup{SVir}} = {\mathbb{C}}[{\partial}] \otimes (\mathbb{C}{L} \oplus \mathbb{C}{G} \oplus \mathbb{C}{L}^{\circ} \oplus \mathbb{C}{G}^{\circ})$ is the Lie conformal algebra endowed with the following $\lambda$-brackets:
\begin{gather*}
[L \ _{\lambda} \ L] = ({\partial} + 2{\lambda})L,
{\quad}
[{L} \ _{\lambda} \ {L}^{\circ}] = ({\partial} + 2{\lambda}){L}^{\circ},
{\quad}
[L \ _{\lambda} \ G] = ({\partial} + {\frac{3}{2}}{\lambda})G,
{\quad}
[L \ _{\lambda} \ {G}^{\circ}] = ({\partial} + {\frac{3}{2}}{\lambda}){G}^{\circ},
\\
[G \ _{\lambda} \ {L}^{\circ}] = - ({\frac{1}{2}}{\partial} + {\frac{3}{2}}{\lambda}){G}^{\circ},
{\quad}
[G \ _{\lambda} \ G] = 2{L},
{\quad}
[G \ _{\lambda} \ {G}^{\circ}] = - 2{L}^{\circ},
{\quad}
[{L}^{\circ} \ _{\lambda} \ {L}^{\circ}] = [{L}^{\circ} \ _{\lambda} \ {G}^{\circ}] = [{G}^{\circ} \ _{\lambda} \ {G}^{\circ}] = 0.
\end{gather*}
In other words, as an $N=2$ SUSY vertex algebra, $V(\widetilde{\textup{SVir}})$ is generated by one element ${G}^{\circ}$ and the $N=2$ $\Lambda$-bracket is given by 
\begin{align*}
[{G}^{\circ} \ _{\Lambda} \ {G}^{\circ}] = - (2{\partial} + 3{\lambda}){G}^{\circ} + 2{\chi}^{1}{L}^{\circ}.
\end{align*}
\end{example}

\begin{example}[Another $N = 2$ SUSY extension of Virasoro vertex algebra] \label{N=2 SUSY extension of Virasoro-2}
\rm
If we apply Corollary \ref{corollary:N=2 extension}, i.e. apply Theorem \ref{theorem:N=1 extension} and Theorem \ref{theorem:N=2 extension}  subsequently,  to the Virasoro vertex algebra, we have an $N = 2$ SUSY extension of the Virasoro vertex algebra which is generated by one element ${\bar{L}}^{\circ}$ endowed with the following $N = 2$ ${\Lambda}$-bracket:
\begin{align*}
[{\bar{L}}^{\circ} \ _{\Lambda} \ {\bar{L}}^{\circ}] = - ({\partial} + 2{\lambda}){\bar{L}^{\circ}}.
\end{align*}
Note that this extended vertex algebra is distinct from the extended vertex algebra $V(\widetilde{\textup{SVir}})$ described in Example \ref{N=2 SUSY extension of Virasoro-1}.
\end{example}

\vskip 1mm

By Corollary \ref{corollary:N=2 extension}, the non-zero ${\lambda}$-brackets of the $N = 2$ SUSY extension of the affine vertex algebra of level $0$ are given by
\begin{equation}\label{eq:N=2 affine brackets}
\begin{gathered}\relax
[{a} \ _{\lambda} \ {b}] = [a, b],
{\quad}
[{\bar{a}} \ _{\lambda} \ {b}] = (-1)^{p(a)}[{a} \ _{\lambda} \ {\bar{b}}] = {\overline{[a, b]}},
{\quad}
[{a}^{\circ} \ _{\lambda} \ {b}] = (-1)^{p(a)}[{a} \ _{\lambda} \ {b}^{\circ}] = {[a, b]}^{\circ},
\\
[{\bar{a}}^{\circ} \ _{\lambda} \ {b}] = [{a} \ _{\lambda} \ {\bar{b}}^{\circ}] = (-1)^{p(a)}[{a}^{\circ} \ _{\lambda} \ {\bar{b}}] = (-1)^{p(a)+1}[{\bar{a}} \ _{\lambda} \ {b}^{\circ}] = {\overline{[a, b]}}^{\circ},
\end{gathered}
\end{equation}
and this vertex algebra is isomorphic to $V^{0}_{N = 2}(\mathfrak{g})$ in Example \ref{example:SUSY affine VA}.
As an $N = 2$ analogue of the superconformal current algebra in Remark \ref{remark:superconformal current}, the following proposition shows that a natural $N = 2$ superconformal extension of $V^{0}_{N = 2}(\mathfrak{g})$ can be obtained by considering it as an $N = 2$ SUSY vertex subalgebra of an $N = 2$ superconformal vertex algebra.

\begin{proposition}\label{proposition:N=2 superconformal extn of affine}
Let $\mathfrak{g}=\bigoplus_{\omega \in \Omega}\mathfrak{g}_{\omega}$ be an $\Omega$-graded Lie superalgebra. Then the $N = 2$ SUSY affine vertex algebra $V^{0}_{N = 2}(\mathfrak{g})$ can be embedded into an $N = 2$ superconformal vertex algebra, where the conformal weight of $a \in {\mathfrak{g}}_{\omega}$ is given by ${\Delta}_{a} = 1 - {\omega}$.
\end{proposition}

\begin{proof}
Let ${\textup{Cur}\mathfrak{g}}_{N=2} = \mathbb{C}[\partial] \otimes ({\mathfrak{g}} \oplus {\bar{\mathfrak{g}}} \oplus {\mathfrak{g}}^{\circ} \oplus {\bar{\mathfrak{g}}}^{\circ})$ be the Lie conformal algebra with the ${\lambda}$-brackets in \eqref{eq:N=2 affine brackets} and consider the following $N = 2$ super-Virasoro Lie conformal algebra
\begin{align*}
{\textup{SVir}}_{N=2} = {\mathbb{C}}[{\partial}] \otimes (\mathbb{C}{L} \oplus \mathbb{C}{{G}^{+}} \oplus \mathbb{C}{{G}^{-}} \oplus \mathbb{C}{J}) \oplus {\mathbb{C}}C
\end{align*}
with the $N = 2$ SUSY structure in Example \ref{example:N=2 SUSY of N=2 superVirasoro}, where $C$ is central and ${\partial}C = 0$.
First of all, define the following ${\lambda}$-brackets
\begin{equation}\label{eq:constraints}
\begin{gathered}\relax
[L \ _{\lambda} \ {a}] = ({\partial} + {\Delta}_{a}{\lambda}){a},
{\quad}
[L \ _{\lambda} \ {\bar{a}}] = ({\partial} + ({\Delta}_{a} - {\frac{1}{2}}){\lambda}){\bar{a}},
\\
[L \ _{\lambda} \ {a}^{\circ}] = ({\partial} + ({\Delta}_{a} - {\frac{1}{2}}){\lambda}){a}^{\circ},
{\quad}
[L \ _{\lambda} \ {\bar{a}}^{\circ}] = ({\partial} + ({\Delta}_{a} - 1){\lambda}){\bar{a}}^{\circ},
\end{gathered}
\end{equation}
where ${\Delta}_{a} = 1-\omega$, for $a \in \mathfrak{g}_{\omega}$.
Then, by solving the SUSY extension formulas with the constraints \eqref{eq:constraints}, we can define the remaining ${\lambda}$-brackets as follows 
\begin{gather*}
[{G}^{+} + {G}^{-} \ _{\lambda} \ {a}] = ({\partial} + (2{\Delta}_{a} - 1){\lambda}){\bar{a}},
{\quad}
[{G}^{+} + {G}^{-} \ _{\lambda} \ {\bar{a}}] = {a},
\\
[{G}^{+} + {G}^{-} \ _{\lambda} \ {a}^{\circ}] = -({\partial} + (2{\Delta}_{a} - 2){\lambda}){\bar{a}}^{\circ},
{\quad}
[{G}^{+} + {G}^{-} \ _{\lambda} \ {\bar{a}}^{\circ}] = - {a}^{\circ},
\\
[{G}^{+} - {G}^{-} \ _{\lambda} \ {a}] = {\sqrt{-1}}({\partial} + (2{\Delta}_{a} - 1){\lambda}){a}^{\circ},
{\quad}
[{G}^{+} - {G}^{-} \ _{\lambda} \ {a}^{\circ}] = {\sqrt{-1}}{a},
\\
[{G}^{+} - {G}^{-} \ _{\lambda} \ {\bar{a}}] = {\sqrt{-1}}({\partial} + (2{\Delta}_{a} - 2){\lambda}){\bar{a}}^{\circ},
{\quad}
[{G}^{+} - {G}^{-} \ _{\lambda} \ {\bar{a}}^{\circ}] = {\sqrt{-1}}{\bar{a}},
\\
[{J} \ _{\lambda} \ {a}] = {\sqrt{-1}}(2{\Delta}_{a} - 2){\lambda}{\bar{a}}^{\circ},
{\quad}
[{J} \ _{\lambda} \ {\bar{a}}] = {\sqrt{-1}}{a}^{\circ},
{\quad}
[{J} \ _{\lambda} \ {a}^{\circ}] = -{\sqrt{-1}}{\bar{a}},
{\quad}
[{J} \ _{\lambda} \ {\bar{a}}^{\circ}] = 0.
\end{gather*}
The Jacobi identity of the $\mathbb{C}[\partial]$-module ${\textup{SVir}}_{N=2} \oplus {\textup{Cur}\mathfrak{g}}_{N=2}$ follows from direct computations. Hence the product of the $N = 2$ super-Virasoro vertex algebra and the $N = 2$ SUSY affine vertex algebra of level $0$
\begin{align*}
V({\textup{SVir}}_{N=2} \oplus {\textup{Cur}\mathfrak{g}}_{N=2}) / \left< C - c \right> \simeq {\textup{SVir}}^{c}_{N=2} \otimes {V}^{0}_{N = 2}(\mathfrak{g})
\end{align*}
is an $N = 2$ superconformal vertex algebra.
Moreover, the two odd derivations
\begin{align*}
{D}^{1} = ({G}^{+} + {G}^{-})_{(0)}, \quad {D}^{2} = -{\sqrt{-1}}({G}^{+} - {G}^{-})_{(0)},
\end{align*}
induce the $N = 2$ SUSY vertex algebra structures of both ${\textup{SVir}}^{c}_{N=2}$ and ${V}^{0}_{N = 2}(\mathfrak{g})$.
\end{proof}

\vskip 1mm

Now, we generalize the construction of Theorem \ref{theorem:N=2 extension} for any $N = n$ SUSY cases. For example, in Remark \ref{remark:N=3 extension}, we present the construction of $N = 3$ SUSY extensions of a given universal enveloping $N = 2$ SUSY vertex algebra. The proof for the arbitrary $N = n$ case is identical to the proofs for $N = 2$ (Theorem \ref{theorem:N=2 extension}) and $N = 3$ (Remark \ref{remark:N=3 extension}) cases in essence.

\begin{remark}\label{remark:N=3 extension}
\rm
For a given $N = 2$ SUSY vertex algebra $V({R}_{N = 2})$, where the $N = 2$ SUSY Lie conformal algebra is given by ${{R}_{N = 2}} = {\mathbb{C}}[{\partial}] \otimes {\textup{span}}_{\mathbb{C}}\{ {u}_{i}, {\bar{u}}_{i}, {u}^{\circ}_{i}, {\bar{u}}^{\circ}_{i} \}_{ i {\in} I }$ with the $N = 2$ supersymmetry ${D}^{1}$ and ${D}^{2}$, we can also construct an $N = 3$ SUSY extension of $V({R}_{N = 2})$ as in the proof of Theorem \ref{theorem:N=2 extension}.
More precisely, we construct an $N = 1$ supersymmetry on the ${\mathbb{C}}[{\partial}]$–module ${\widetilde{R}}_{N = 2} := {{R}_{N = 2}} \oplus {R_{N = 2}}^{\dagger}$ in the ${D}^{\dagger}$–direction, where the ${\mathbb{C}}[{\partial}]$-module ${{R}_{N = 2}}^{\dagger} = {\mathbb{C}}[{\partial}] \otimes {\textup{span}}_{\mathbb{C}}\{ {u}^{\dagger}_{i}, {\bar{u}}^{\dagger}_{i}, {u}^{\circ, \dagger}_{i}, {\bar{u}}^{\circ, \dagger}_{i} \}_{ i {\in} I }$ is given by Theorem \ref{theorem:N=1 extension}.
If we define the three odd endomorphisms ${D}^{1}$, ${D}^{2}$ and ${D}^{3} = {D}^{\dagger}$ as shown in the diagram below:
\[
\begin{tikzcd}[style={font=\small}]
& {\bar{u}}_{i} \arrow[rr, "{D}^{1}"] \arrow[from=ddd, "{D}^{3}"] & & {u}_{i} \\
{\bar{u}}^{\circ}_{i} \arrow[crossing over, rr, "\qquad -{D}^{1}"] \arrow[ru, "{D}^{2}"] & & {u}^{\circ}_{i} \arrow[ru, "{D}^{2}"] & \\ \\
& {\bar{u}}^{\dagger}_{i} \arrow[rr, "\qquad -{D}^{1}"] & & {u}^{\dagger}_{i} \arrow[uuu, "{D}^{3}"] \\
{{\bar{u}}^{\circ, \dagger}_{i}} \arrow[rr, swap, "{{D}^{1}}"] \arrow[uuu, "{D}^{3}"] \arrow[ur, "-{D}^{2}"] & & {{u}^{\circ, \dagger}_{i}} \arrow[ru, swap, "-{D}^{2}"] \arrow[uuu, crossing over, "{D}^{3}"] &
\end{tikzcd}
\]
\noindent then each of the following pairs satisfies the SUSY extension formulas:
\begin{align*}\label{eq:N=3 SUSY generator}
(\{ {\bar{u}}_{i}, {\bar{u}}^{\circ}_{i}, {\bar{u}}^{\dagger}_{i}, {\bar{u}}^{\circ, \dagger}_{i} \}_{i {\in} I}, {D}^{1}),
{\quad}
(\{ {u}^{\circ}_{i}, {\bar{u}}^{\circ}_{i}, {u}^{\circ, \dagger}_{i}, {\bar{u}}^{\circ, \dagger}_{i} \}_{i {\in} I}, {D}^{2}),
{\quad}
(\{ {u}^{\dagger}_{i}, {\bar{u}}^{\dagger}_{i}, {u}^{\circ, \dagger}_{i}, {\bar{u}}^{\circ, \dagger}_{i} \}_{i {\in} I}, {D}^{3}).
\end{align*}
Indeed, the non-trivial part of the proof is the verification of the SUSY extension formulas for the pairs:
\begin{align*}
(\{ {\bar{u}}_{i}, {\bar{u}}^{\circ, \dagger}_{i} \}_{i {\in} I}, {D}^{1}),
{\quad}
(\{ {\bar{u}}^{\dagger}_{i}, {\bar{u}}^{\circ}_{i} \}_{i {\in} I}, {D}^{1}),
{\quad}
(\{ {u}^{\circ}_{i}, {\bar{u}}^{\circ, \dagger}_{i} \}_{i {\in} I}, {D}^{2}),
{\quad}
(\{ {u}^{\circ, \dagger}_{i}, {\bar{u}}^{\circ}_{i} \}_{i {\in} I}, {D}^{2}),
\end{align*}
which follows from the fact that $(\{{\bar{u}}_{i}, {\bar{u}}^{\circ}_{i}\}_{i {\in} I}, {D}^{1})$ and $(\{{u}^{\circ}_{i}, {\bar{u}}^{\circ}_{i}\}_{i {\in} I}, {D}^{2})$ satisfy the SUSY extension formulas. For example, we have the (SEF1) of the pair $(\{ {\bar{u}}_{i} \}_{i {\in} I}, \{ {\bar{u}}^{\circ, \dagger}_{i} \}_{i {\in} I})$ in the ${D}^{1}$–direction from the following computations:
\begin{align*}
[{D}^{1}{\bar{u}}_{i} \ _{\lambda} \ {D}^{1}{\bar{u}}^{\circ, \dagger}_{j}] 
&= [{u}_{i} \ _{\lambda} \ {u}^{\circ, \dagger}_{j}] 
= (-1)^{p({u}_{i})}[{u}_{i} \ _{\lambda} \ {u}^{\circ}_{j}]^{\dagger} 
= - ({D}^{1}[{u}_{i} \ _{\lambda} \ {\bar{u}}^{\circ}_{j}])^{\dagger} - {\lambda}[{\bar{u}}_{i} \ _{\lambda} \ {\bar{u}}^{\circ}_{j}]^{\dagger}
\\
&= {D}^{1}([{u}_{i} \ _{\lambda} \ {\bar{u}}^{\circ}_{j}]^{\dagger}) - {\lambda}[{\bar{u}}_{i} \ _{\lambda} \ {\bar{u}}^{\circ}_{j}]^{\dagger} 
= (-1)^{p({u}_{i})}({D}^{1}[{u}_{i} \ _{\lambda} \ {\bar{u}}^{\circ, \dagger}_{j}] + {\lambda}[{\bar{u}}_{i} \ _{\lambda} \ {\bar{u}}^{\circ, \dagger}_{j}]).
\end{align*}
\end{remark}

\vskip 1mm

Before the closing of this paper, we suggest additional systematic methods to find SUSY structures of vertex algebras (see Remark \ref{remark:non-univ N2 SUSY}, \ref{remark:non-univ N3,4 SUSY} and \ref{remark:general VA embedding}). These methods can be used for a vertex algebra which is not a universal enveloping vertex algebra.

\begin{remark}\label{remark:non-univ N2 SUSY}
\rm
Let $V$ be a vertex algebra with an $N = 1$ superconformal vector $G$. Suppose  there is a decomposition $G = {G}^{+} + {G}^{-}$ satisfying $({{G}^{\pm}}_{(0)}{G}^{\pm})_{(0)} = 0$. Then $V$ has an $N = 2$ SUSY vertex algebra structure. Indeed, as in the case of $N = 2$ superconformal vertex algebra, define the odd derivations as
\begin{align*}
{D}^{1} = ({\mu}{G}^{+} + {\mu}^{-1}{G}^{-})_{(0)}, \quad {D}^{2} = {\sqrt{-1}}({\mu}{G}^{+} - {\mu}^{-1}{G}^{-})_{(0)},
\end{align*}
for ${\mu} \in {\mathbb{C}}^{*}$. Then, by the Jacobi identity, we can check $[D^{i}, D^{j}] = 2{\delta}_{ij}{\partial}$ for $i, j = 1, 2$. Also, it follows from Proposition 1.12 of \cite{DSK06} that ${D}^{1}$ and ${D}^{2}$ are odd derivations of the normally ordered product and the ${\lambda}$-brackets of $V$.
\end{remark}

\begin{remark}\label{remark:non-univ N3,4 SUSY}
\rm
A vertex algebra is called an \textit{$N = 3$ superconformal vertex algebra} if it has a conformal vector $L$ with central charge $c$, three even primary vectors ${J}^{+}$, ${J}^{-}$, ${J}^{0}$ of conformal weight $1$, three odd primary vectors ${G}^{+}$, ${G}^{-}$, ${G}^{0}$ of conformal weight $\frac{3}{2}$, and an odd primary vector ${\Phi}$ of conformal weight $\frac{1}{2}$, satisfying the $N = 3$ super-Virasoro relation \eqref{eq:N=3 superVirasoro relation}. For an $N = 3$ superconformal vertex algebra, there is an $N = 3$ SUSY vertex algebra structure with the following odd derivations 
\begin{align*}
{D}^{1} = ({\mu}{G}^{+} + {\mu}^{-1}{G}^{-})_{(0)},
{\quad}
{D}^{2} = {\sqrt{-1}}({\mu}{G}^{+} - {\mu}^{-1}{G}^{-})_{(0)},
{\quad}
{D}^{3} = {\sqrt{2}{G}^{0}}_{(0)}.
\end{align*}
For an $N = 4$ superconformal vertex algebra, if we adopt the definition and notation from Section 8.4 of \cite{KW04}, the following odd derivations
\begin{gather*}
{D}^{1} = ({\mu}{G}^{+} + {\mu}^{-1}{\bar{G}}^{-})_{(0)},
{\quad}
{D}^{2} = {\sqrt{-1}}({\mu}{G}^{+} - {\mu}^{-1}{\bar{G}}^{-})_{(0)},
\\
{D}^{3} = ({\nu}{G}^{-} + {\nu}^{-1}{\bar{G}}^{+})_{(0)},
{\quad}
{D}^{4} = {\sqrt{-1}}({\nu}{G}^{-} - {\nu}^{-1}{\bar{G}}^{+})_{(0)},
\end{gather*}
give rise to an $N = 4$ SUSY vertex algebra structure, for $\mu, \nu \in {\mathbb{C}}^{*}$.
\end{remark}

\begin{remark}\label{remark:general VA embedding}
\rm
For any vertex algebra $V$, it can be embedded into SUSY vertex algebras by considering the universal enveloping vertex algebra of the vertex algebra $V$.
Therefore, we know that there is an $N = n'$ SUSY vertex algebra containing a given $N = n$ SUSY vertex algebra, for any $n' \geq n$.
\end{remark}

\end{document}